\documentclass[a4paper, 12pt]{article}

\usepackage[utf8]{inputenc}

\RequirePackage{amsthm,amsmath,amsfonts}

\usepackage[backend=biber, natbib=true, style=authoryear, doi=true]{biblatex}
\RequirePackage[colorlinks,citecolor=blue,urlcolor=blue]{hyperref}
\RequirePackage{graphicx}
\usepackage{mathtools}
\usepackage{csquotes}
\usepackage{enumitem}
\usepackage{hyperref}
\usepackage{bbm}
\usepackage[capitalize,nameinlink,noabbrev]{cleveref}

\usepackage[caption=false]{subfig}
\usepackage{array}
\usepackage{tikz}
\usetikzlibrary{calc}
\usepackage{multicol}
\usepackage{xspace}
\usepackage[caption=false]{subfig}
\usepackage{authblk}
\usepackage[linesnumbered,ruled,lined,boxed,commentsnumbered]{algorithm2e}
\usepackage{todonotes} 
\setuptodonotes{color=blue!30}
\usepackage{comment}
\usepackage{orcidlink}

\usepackage{booktabs} 

\numberwithin{equation}{section}

\theoremstyle{plain}

\newtheorem{theorem}{Theorem}[section]
\newtheorem{lemma}[theorem]{Lemma}
\newtheorem{corollary}[theorem]{Corollary}

\theoremstyle{remark}

\newtheorem{definition}{Definition}

\newtheorem{assumptions}{Assumption}
\crefname{assumptions}{Assumption}{Assumptions}
\crefname{algocfline}{Algorithm}{Algorithms}

\newcommand{\xistref}{\hyperref[alg:xist]{Xist}\xspace}
\newcommand{\multixistref}{\hyperref[alg:xist_iterated]{Multi-Xist}\xspace}
\newcommand{\xvstnameref}{\hyperref[alg:xvst]{basic Xvst algorithm}\xspace}
\newcommand{\xistnameref}{\hyperref[alg:xist]{Xist algorithm}\xspace}

\DeclareMathOperator*{\argmin}{arg\,min}

\DeclareMathOperator{\bal}{Bal}

\DeclareMathOperator{\XC}{\mathrm{XC}}
\DeclareMathOperator{\MC}{\mathrm{MC}}
\DeclareMathOperator{\RC}{\mathrm{RC}}
\DeclareMathOperator{\NC}{\mathrm{NC}}
\DeclareMathOperator{\CC}{\mathrm{CC}}
\DeclareMathOperator{\kXC}{\mathrm{kXC}}
\DeclareMathOperator{\kMC}{\mathrm{kMC}}

\DeclareMathOperator{\xist}{\mathrm{Xist}}
\DeclareMathOperator{\Vloc}{V^{\mathrm{loc}}}

\DeclareMathOperator{\vol}{vol}

\DeclareMathOperator{\minn}{min}

\renewcommand{\vec}[1]{{\boldsymbol#1}} 
\newcommand{\mat}[1]{{\boldsymbol#1}}    
\newcommand{\R}{\mathbb{R}} 
\newcommand{\N}{\mathbb{N}} 
\newcommand{\BO}{\mathcal{O}} 
\newcommand{\eps}{\varepsilon} 

\newcommand{\abs}[1]{|#1|} 
\newcommand{\comp}[1]{{#1}^{\mathsf{c}}}
\newcommand{\norm}[1]{\left\lVert#1\right\rVert}
\newcommand{\mysubref}[2]{\hyperref[#1:#2]{\cref*{#1}~\ref*{#1:#2}}}
\newcommand{\assumptionref}[1]{\hyperref[#1]{Assumption \ref*{#1}}}
\renewcommand{\eqref}[1]{\hyperref[#1]{(\ref*{#1})}}

\newcolumntype{x}[1]{>{\centering\arraybackslash\hspace{0pt}}p{#1}}

\definecolor{orange1}{RGB}{239, 129, 27}
\definecolor{rblue}{rgb}{0,0,255}
\definecolor{rgreen3}{rgb}{0,205,0}
\definecolor{niceblue}{rgb}{0.15,0.15,1}
\definecolor{nicegreen}{rgb}{0.1,0.9,0.1}

\begin{document}

\title{A scalable clustering algorithm to approximate graph cuts}
\date{April 10, 2024}

\author{Leo Suchan\orcidlink{0000-0001-7212-5492}, Housen Li\orcidlink{0000-0002-2434-9878}, and Axel Munk\orcidlink{0000-0002-9181-9331}
\thanks{The authors are with the Institute of Mathematical Stochastics, Georg August University of Göttingen, Göttingen, Germany}
\thanks{H.\ Li and A.\ Munk are with the Cluster of Excellence \enquote{Multiscale Bioimaging: from Molecular Machines to Networks of Excitable Cells} (MBExC)}
}

\maketitle

\begin{abstract}
    Due to their computational complexity, graph cuts for cluster detection and identification are used mostly in the form of convex relaxations. We propose to utilize the original graph cuts such as Ratio, Normalized or Cheeger Cut to detect clusters in weighted undirected graphs by restricting the graph cut minimization to $st$-MinCut partitions. Incorporating a vertex selection technique and restricting optimization to tightly connected clusters, we combine the efficient computability of $st$-MinCuts and the intrinsic properties of Gomory-Hu trees with the cut quality of the original graph cuts, leading to linear runtime in the number of vertices and quadratic in the number of edges. Already in simple scenarios, the resulting algorithm Xist is able to approximate graph cut values better empirically than spectral clustering or comparable algorithms, even for large network datasets. We showcase its applicability by segmenting images from cell biology and provide empirical studies of runtime and classification rate.
\end{abstract}

\section{Introduction}
\label{sec:intro}

The detection and identification of clusters in datasets is a fundamental task of data analysis, with applications including image segmentation \citep{SnethilnathSindhuOmkar2014, vandenHeuvelMandlHulshoff2008, Wang_etal2020}, machine learning \citep{ChewCahill2015, Tang_etal2016, MalioutovBarzilay2006}, and parallel computing \citep{YuShenHu2016, Peng_etal2014, Chen_etal2011}. The construction and partitioning of a graph is key to many popular methods of cluster detection. Prominent graph cuts such as Ratio Cut \cite{HagenKahng1992}, Normalized Cut \cite{ShiMalik2000} or Cheeger Cut \cite{Cheeger1971} are designed to partition a graph in a \enquote{balanced} way while providing an intuitive geometric interpretation. Ratio Cut, for instance, attempts to balance cluster sizes, and Normalized Cut and Cheeger Cut aim for equal volumes of the resulting partitions. However, as the computation of the above mentioned cuts is an NP-hard problem (see e.g.\ \cite{Mohar1989, BuiJones1992, ShiMalik2000, SimaSchaeffer2006}), attention has shifted primarily to their various convex relaxations and regularizations. Most prominent are spectral clustering techniques, which can be viewed as a convex relaxation of the optimization problem underlying Normalized and Cheeger Cut \cite{vonLuxburg2007}. They offer a quick and simple way to partition a graph, with a worst-case runtime cubic in the number of vertices \cite{vonLuxburg2007}, and their practical success has been demonstrated in numerous applications. Consequently, they have been the subject of extensive theoretical analysis \citep{LuxburgBelkinBousqet2008, MaierLuxburgHein2013, TrillosSlepcev2018} as well as the inspiration for a multitude of specialized algorithms \citep{NgJordanWeiss2001, NascimentoCarvalho2011, ZhongPun2022}. However, it is well known that spectral clustering does not always yield a qualitatively sensible partition, the most famous example being the so-called \enquote{cockroach graph} and its variants \cite{GuatteryMiller1998}. More significantly, it has been shown that spectral clustering possesses fundamental flaws in detecting clusters of different scales \cite{NadlerGalun2007}.

To overcome these issues we propose a different approach which combines the versatility and quality of graph cuts with the most significant trait of spectral clustering, its quick and simple computability. While spectral clustering builds on Normalized Cut and Cheeger Cut alone, the suggested algorithm \emph{\xistref} is applicable to any balanced graph cut functional (also sometimes called sparsest cut), thus making it adaptable and scalable. In contrast to spectral clustering, our algorithm approximates graph cuts not by means of relaxation of the graph cut functional, but instead by restricting minimization onto a particularly designed subset of partitions. This subset is constituted by $st$-MinCuts, where $s$ and $t$ runs through a certain subset of vertices. These partitions can be computed very fast through max flows~\citep{Orlin2013} and the well-known duality of the max-flow and min-cut problems. Hence, our proposed algorithm \xistref to a large extent preserves the combinatorial nature of the problem while retaining a worst computational complexity \emph{quadratic} in the number of vertices and \emph{linear} in the number of edges (\cref{lem:xist_complexity}).

The rest of the paper is organized as follows. \cref{sec:notation} introduces the basic notation. The proposed 2-way cut algorithm is detailed in \cref{sec:algorithms}, together with its theoretical properties. In \cref{sec:simulations}, we study the performance of the proposed algorithm on simulated and real world datasets and introduce a multiway cut extension. \cref{sec:outlook} provides a discussion of further extensions and concludes the paper. Technical proofs are given in the \nameref{appendix}. 

\section{Definitions and notation}
\label{sec:notation}

\noindent We consider simple, undirected, and weighted graphs, denoted as $G = (V, E, \mat{W})$, with $V$ the vertex set, $E$ the set of edges and $\mat{W}=(w_{ij})_{i,j\in V}$ the weight matrix. Each entry $w_{ij}$ equals the weight of edge $\{i,j\}$ if $\{i,j\} \in E$ and zero otherwise. Thus, $\mat{W}$ is symmetric and has a zero diagonal.

\begin{definition}[Graph cut]
	\label{def:cuts}
	For a simple, undirected, weighted graph $G=(V,E,\mat{W})$, we define the \emph{(balanced) graph cut} (or: \emph{sparsest cut}) $\XC$ of $G$ as:
	\begin{align*}
	\XC(G) &\coloneqq \min_{S\subset V}\XC_S(G)\\\quad\text{where}\quad\XC_S(G) &\coloneqq \sum_{{i\in S, j\in \comp{S}}}\frac{w_{ij}}{\bal_G(S,\comp{S})}
	\end{align*}
	for $S\subset V$. Here, $\comp{S}\coloneqq V\setminus S$ is the complement of $S$, and $\XC$ and $\bal_G(S,\comp{S})$ serve as placeholders for the graph cut (see \cref{tbl:cuts}) and its corresponding balancing term, respectively.
\end{definition}

\cref{tbl:cuts} lists the balancing terms for Minimum Cut (MinCut), Ratio Cut, Normalized Cut (NCut) and Cheeger Cut. In general, these balancing terms can depend on the underlying graph structure, the partition $S$ as well as the weight matrix $\mat{W}$. For any partition $S\subset V$, define its \emph{size} as $\abs{S} \coloneqq  \#\{i\in S\}$ and its \emph{volume} as $\vol(S)\coloneqq \sum_{i\in S}\deg(i)$, where $\deg(i)\coloneqq\sum_{j\in V} w_{ij}$ is the \emph{degree} of vertex $i$. Assume $n\coloneqq \abs{V} < \infty$ and $m\coloneqq \abs{E}< \infty$.

\begin{table}[htpb]
    \centering
    \caption{Definitions of common (balanced) graph cuts.}
	\label{tbl:cuts}
	\begin{tabular}{llll}
			\toprule
			Cut name & $\XC$ & $\bal_G(S,\comp{S})$ & Reference\\
			\midrule
			Minimum Cut & $\MC$ & $1$ & E.g.~\cite{Cook_etal1998}\\
			Ratio Cut & $\RC$ & $\abs{S}\,\abs{\comp{S}}$ & \cite{HagenKahng1992}\\
			Normalized Cut & $\NC$ & $\vol(S)\vol(\comp{S})$ & \cite{ShiMalik2000}\\
			Cheeger Cut & $\CC$ & $\min\{\vol(S),\,\vol(\comp{S})\}$ & \cite{Cheeger1971}\\
			\bottomrule
	\end{tabular}
\end{table}

The disadvantage of (balanced) graph cuts is their computational complexity that is rooted in their combinatorial nature. It has been shown that the problem of computing NCut is NP-complete, see \cite[Appendix~A, Proposition~1]{ShiMalik2000}, and the idea behind the proof can be adapted for the other balanced cuts in \cref{tbl:cuts}, namely Ratio Cut and  Cheeger Cut (see \cite{SimaSchaeffer2006} for an alternative proof for the latter). Several other balanced graph cuts are also known to be NP-complete to compute, for instance, the graph cut with the balancing term $\bal_G(S,\comp{S})=\min\{|S|,|\comp{S}|\}$, see \cite{Mohar1989}. It should be noted that, due to its lack of a balancing term, MinCut is computable in polynomial time (see e.g.\ \cite{GawrychowskiMozesWeimann2024} for a randomized algorithm with high probability in $\BO(m\log^2 n)$ time). For applications, however, MinCut is of limited relevance since it tends to separate a single vertex from the remainder of the graph \cite{vonLuxburg2007}. Thus, all practically relevant graph cuts become non-computable even for problems of moderate sizes. This also remains the case for many relaxed versions \cite{WagnerWagner1993} and approximations of graph cuts up to a constant factor \cite{BuiJones1992}. \citet{Sherman2009} proposed an algorithm to approximate Ratio Cut on an unweighted graph up to a $\BO(\sqrt{\log n})$ factor in $\widetilde{\BO}(n\cdot\mathrm{polylog}(n))$ time, building upon previous results by \cite{KhandekarRaoVazirani2009,AroraKale2007}, and attain the lower bound on the approximation factor as shown by \cite{Orecchia_etal2008}. To the best of our knowledge, however, these results do not extend to other balancing terms or to weighted graphs. As mentioned in the \nameref{sec:intro}, spectral clustering represents a notable exception as it can be regarded as a relaxation of Normalized and Cheeger Cut while still being computable in $\BO(n^3)$ time for weighted graphs.

\section{The algorithms}
\label{sec:algorithms}

\noindent We first introduce the basic algorithm and then present a refined and vastly accelerated variant.

\subsection{A basic algorithm for imitating graph cuts through \texorpdfstring{$st$}{st}-min cuts}
\label{sub:xvst}

\noindent To retain the qualitative aspects of the (balanced) graph cuts themselves as much as possible we suggest to restrict the combinatorial optimization to a certain collection of partitions. Then, if such a collection of partitions is well-chosen, the original graph cut (i.e.\ the minimizer over all partitions) can be imitated on a qualitative level. For this purposes, we consider some collections of $st$-MinCut partitions, i.e.\ the cuts that separate the two nodes $s$ and $t$ in $V$ for $s,t\in V$. More precisely, an \emph{$st$-MinCut partition} $S_{st}$ is defined as 
\begin{align*}
S_{st} \; &\in\;\argmin_{S\in\mathcal{S}_{st}^*} \MC_S(G),\\
\quad\text{where}\quad\mathcal{S}_{st}^* &\coloneqq  \{S \subset V\; :\; s \in S,\ t\not\in S\}.
\end{align*}

The partition $S_{st}$ might not be unique. However, the nonuniqueness represents a fringe case that is highly unlikely to occur on real-world data; see \cref{sec:simulations}. Also notice that, by definition, $s\in S_{st}$ and $t\in\comp{S}_{st}$ for any $s,t\in V$, $s\neq t$. This property will be important later. The fastest algorithms for computing an $st$-MinCut partition take $\BO(nm)$ time for general graphs; for instance, one could use the algorithm suggested in \citet{Orlin2013} if $m = \BO(n^{1.06})$, and \citet{OrlinGong2021} otherwise, see also \cref{apdx:xist_complexity_proof}.

Our idea is to restrict the graph cut minimization to the $st$-MinCut partitions. The reason for choosing this particular subset of partitions is twofold: First, it is computable in polynomial time as outlined above, and second, an $st$-MinCut partition is forced to separate two clusters if $s$ and $t$ belong to different clusters that are more tightly connected than the edges connecting them. In this scenario, the MinCut value of cutting between the two clusters is lower than cutting off parts of any one of the clusters as the cut still needs to separate $s$ and $t$. This requirement of the two clusters being more intraconnected than interconnected is precisely how one would define a cluster, so for proper choice of vertices $s$ and $t$ one expects the $st$-MinCut partition to be \enquote{reasonable} (i.e.\ in that it separates two clusters). A visualization of this can be seen in \cref{img:xist_explained} where one could consider the orange vertices to be good choices for $s$ and $t$ -- this is elaborated upon in \cref{sub:xist}.

A first version of our algorithm can be stated as follows:

\begin{algorithm}
\label{alg:xvst}
\SetAlgorithmName{Basic Xvst algorithm}{Xvst}
\SetAlgoLined
\DontPrintSemicolon
\SetKwInOut{Input}{input}
\SetKwInOut{Output}{output}
\Input{weighted graph $G=(V,E,\mat{W})$}
\Output{XC value $c_{\min}$ with associated partition $S_{\min}$}
\BlankLine

set $c_{\min}\leftarrow\infty$ and $S_{\min}\leftarrow\emptyset$\;

\For{$\{s,t\}\subset V$ with $s\neq t$\label{alg:xvst:while_loop}}{
    compute an $st$-MinCut partition $S_{st}$ on $G$ \label{alg:xvst:st-mincut}\;
    compute the graph cut value $\XC_{\comp{S}_{st}}(G)$ of partition $S_{st}$ \label{alg:xvst:xcut}\;
	\If{$\XC_{S_{st}}(G) < c_{\min}$}{
        $c_{\min} \leftarrow \XC_{S_{st}}(G)$\;
        $S_{\min} \leftarrow S_{st}$\;
    }
}\label{alg:xvst:while_end}
\caption{$\XC(G)$ via $st$-MinCuts}
\end{algorithm}

As the computation of $st$-MinCut in line~\ref{alg:xvst:st-mincut} of our \xvstnameref is in $\BO(nm)$ time, the complexity of this algorithm is $\BO(n^3 m)$ (\cref{lem:xist_complexity}). While this is not particularly fast, the design of the algorithm guarantees that the resulting partition is a cut that is reasonable in the sense that it separates two vertices $s$ and $t$ through the $st$-MinCut partition $S_{st}$ while also taking cluster size into account via the balancing term in the graph cut value $\XC_{S_{st}}({G})$. This imitates the nature of (the NP-complete problem of) computing graph cuts from a qualitative perspective, whereas techniques such as spectral clustering approximate the corresponding functionals by convex relaxations. More precisely, the partition $S_{\minn}$ that the \xvstnameref outputs is guaranteed to be an $st$-MinCut for some $s,t\in V$, whereas the partition returned by spectral clustering does not possess any inherent qualitative feature per se.

In the literature, the attention has mainly focused on the set $\mathcal{S}^*_{st}$ of $st$-MinCut partitions for a fixed pair of $s,t\in V$, $s\neq t$. For instance, the cardinality of $\mathcal{S}^*_{st}$ has been used as a structure characterization on the {crossing minimization problem} in graph planarizations \cite{ChimaniGutwengerMutzel2007}. \citet{AndersenLang2008} proposed an algorithm subroutine to improve existing partitions that is based on $st$-MinCuts; their method, however, introduces artificial vertices $s$ and $t$ to act as penalization for changing the existing partition, in difference to our algorithm which considers only \enquote{real} vertices $s,t\in\Vloc$. It should be noted that \citet{Bonsma2010} showed that the problem of finding the most balanced partition in $\mathcal{S}^*_{st}$ is NP-hard, so that we consider only one $st$-MinCut partition $S_{st}$ for a given pair of $s,t\in V$, and employ instead the collection of such partitions for all pairs of $s,t\in V$, namely, $\{S_{st} \colon s,t \in V, s\neq t\}$. In this way, we preserve the intrinsic structure of the graph to a large extend while gaining efficient computation in polynomial time.

\subsection{The proposed \texorpdfstring{\xistnameref}{Xist algorithm}}
\label{sub:xist}

\begin{figure*}[!t]
    \centering
    \subfloat[All \emph{local maxima} $1$, $2$, $3$, and $4$.\\ Initially, $\tau = (1,1,1,1)$.]{
			\centering
			\begin{tikzpicture}[x=8mm,y=8mm,rotate=180]
			\coordinate (O) (0,0);
            \node[inner sep=0pt] (A) at (1.2,6.1) {A};
            \node[inner sep=0pt] (B) at (2.0,1.3) {B};
            \node[inner sep=0pt] (C) at (2.7,2.1) {C};
            \node[inner sep=0pt] (D) at (1.5,2.4) {D};
            \node[inner sep=0pt] (E) at (2.6,1.2) {E};
            \node[inner sep=0pt] (F) at (2.2,4.8) {F};
            \node[inner sep=0pt] (G) at (4.2,2.3) {G};
            \node[inner sep=0pt] (H) at (2.7,3.1) {H};
            \node[inner sep=0pt] (I) at (6.6,2.2) {I};
            \node[inner sep=0pt] (J) at (3.2,4.3) {J};
            \node[inner sep=0pt] (K) at (5.5,1.9) {K};
            \node[inner sep=0pt] (L) at (7.2,3.5) {L};
            \node[inner sep=0pt] (M) at (7.8,5.0) {M};
            \node[inner sep=0pt] (N) at (8.7,2.5) {N};
            \node[inner sep=0pt] (O) at (8.4,6.4) {O};
            \node[inner sep=0pt] (P) at (8.2,1.4) {P};
            \node[inner sep=0pt] (Q) at (7.7,2.2) {Q};
            \node[inner sep=0pt] (R) at (8.2,3.4) {R};
            \node[inner sep=0pt] (S) at (6.1,4.7) {S};
            \node[inner sep=0pt] (T) at (6.2,6.2) {T};
            \node[inner sep=0pt] (U) at (4.1,4.6) {U};
            \node[inner sep=0pt] (V) at (5.0,3.6) {V};
            \node[inner sep=0pt] (W) at (3.6,5.7) {W};
            \node[inner sep=0pt] (X) at (7.3,5.8) {X};

            \draw[black!50, line width=0.72mm] (A) -- (F) -- (J) -- (W) -- (U) -- (V) node[pos=0.5, inner sep=0pt, outer sep=0pt] (UV) {} -- (S) node[pos=0.5, inner sep=0pt, outer sep=0pt] (SV) {} -- (X) -- (T) -- (X) -- (O) -- (M) -- (X);
            \draw[black!50, line width=0.72mm] (L) -- (Q) -- (R) -- (N) -- (P) -- (Q) -- (I) -- (K) -- (G) node[pos=0.5, inner sep=0pt, outer sep=0pt] (KG) {} -- (C) -- (E) -- (B) -- (C) -- (D);
            \draw[black!50, line width=0.72mm] (B) -- (D) -- (H) -- (C);
            \draw[black!50, line width=0.72mm] (G) -- (V);
            \draw[black!50, line width=0.72mm] (L) -- (I) -- (P);
            \draw[black!50, line width=0.72mm] (S) -- (T);
            \draw[black!50, line width=0.72mm] (N) -- (Q);

            \draw[black!50, line width=0.9mm] (H) -- (D) -- (C) -- (H) -- (J) node[pos=0.5, inner sep=0pt, outer sep=0pt] (JH) {} -- (U) -- (W) -- (F) -- (J);
            \draw[black!50, line width=0.9mm] (S) -- (M) -- (X) -- (S) -- (V);
            \draw[black!50, line width=0.9mm] (I) -- (Q) -- (P) -- (N) -- (Q);
            
            \draw[black!50, line width=1.08mm] (S) -- (L) node[pos=0.5, inner sep=0pt, outer sep=0pt] (SL) {} -- (Q) -- (R) -- (L);
            \draw[black!50, line width=1.08mm] (C) -- (G);
            \draw[black!50, line width=1.08mm] (I) -- (K);
            \draw[black!50, line width=1.08mm] (C) -- (E) -- (B) -- (C);
            
            \foreach \i in {A,B,...,X}{
                \filldraw (\i) circle (6pt);
            }
            \filldraw[orange] (C) circle (6pt) node[white] {\textbf{3}};
            \filldraw[orange] (J) circle (6pt) node[white] {\textbf{4}};
            \filldraw[orange] (Q) circle (6pt) node[white] {\textbf{2}};
            \filldraw[orange] (S) circle (6pt) node[white] {\textbf{1}};
			\end{tikzpicture}}%
%
		\subfloat[$21$-MinCut with $\RC\approx 0.021$.\\ Updated $\tau = (1,2,1,1)$]{
			\centering
			\begin{tikzpicture}[x=8mm,y=8mm,rotate=180]
			\coordinate (O) (0,0);
            \node[inner sep=0pt] (A) at (1.2,6.1) {A};
            \node[inner sep=0pt] (B) at (2.0,1.3) {B};
            \node[inner sep=0pt] (C) at (2.7,2.1) {C};
            \node[inner sep=0pt] (D) at (1.5,2.4) {D};
            \node[inner sep=0pt] (E) at (2.6,1.2) {E};
            \node[inner sep=0pt] (F) at (2.2,4.8) {F};
            \node[inner sep=0pt] (G) at (4.2,2.3) {G};
            \node[inner sep=0pt] (H) at (2.7,3.1) {H};
            \node[inner sep=0pt] (I) at (6.6,2.2) {I};
            \node[inner sep=0pt] (J) at (3.2,4.3) {J};
            \node[inner sep=0pt] (K) at (5.5,1.9) {K};
            \node[inner sep=0pt] (L) at (7.2,3.5) {L};
            \node[inner sep=0pt] (M) at (7.8,5.0) {M};
            \node[inner sep=0pt] (N) at (8.7,2.5) {N};
            \node[inner sep=0pt] (O) at (8.4,6.4) {O};
            \node[inner sep=0pt] (P) at (8.2,1.4) {P};
            \node[inner sep=0pt] (Q) at (7.7,2.2) {Q};
            \node[inner sep=0pt] (R) at (8.2,3.4) {R};
            \node[inner sep=0pt] (S) at (6.1,4.7) {S};
            \node[inner sep=0pt] (T) at (6.2,6.2) {T};
            \node[inner sep=0pt] (U) at (4.1,4.6) {U};
            \node[inner sep=0pt] (V) at (5.0,3.6) {V};
            \node[inner sep=0pt] (W) at (3.6,5.7) {W};
            \node[inner sep=0pt] (X) at (7.3,5.8) {X};

            \draw[black!50, line width=0.72mm] (A) -- (F) -- (J) -- (W) -- (U) -- (V) node[pos=0.5, inner sep=0pt, outer sep=0pt] (UV) {} -- (S) node[pos=0.5, inner sep=0pt, outer sep=0pt] (SV) {} -- (X) -- (T) -- (X) -- (O) -- (M) -- (X);
            \draw[black!50, line width=0.72mm] (L) -- (Q) -- (R) -- (N) -- (P) -- (Q) -- (I) -- (K) -- (G) node[pos=0.5, inner sep=0pt, outer sep=0pt] (KG) {} -- (C) -- (E) -- (B) -- (C) -- (D);
            \draw[black!50, line width=0.72mm] (B) -- (D) -- (H) -- (C);
            \draw[black!50, line width=0.72mm] (G) -- (V);
            \draw[black!50, line width=0.72mm] (L) -- (I) -- (P);
            \draw[black!50, line width=0.72mm] (S) -- (T);
            \draw[black!50, line width=0.72mm] (N) -- (Q);

            \draw[black!50, line width=0.9mm] (H) -- (D) -- (C) -- (H) -- (J) node[pos=0.5, inner sep=0pt, outer sep=0pt] (JH) {} -- (U) -- (W) -- (F) -- (J);
            \draw[black!50, line width=0.9mm] (S) -- (M) -- (X) -- (S) -- (V);
            \draw[black!50, line width=0.9mm] (I) -- (Q) -- (P) -- (N) -- (Q);
            
            \draw[black!50, line width=1.08mm] (S) -- (L) node[pos=0.5, inner sep=0pt, outer sep=0pt] (SL) {} -- (Q) -- (R) -- (L);
            \draw[black!50, line width=1.08mm] (C) -- (G);
            \draw[black!50, line width=1.08mm] (I) -- (K);
            \draw[black!50, line width=1.08mm] (C) -- (E) -- (B) -- (C);
            
            \foreach \i in {A,B,...,X}{
                \filldraw[niceblue] (\i) circle (6pt);
            }
            \filldraw[nicegreen] (K) circle (6pt);
            \filldraw[nicegreen] (I) circle (6pt);
            \filldraw[nicegreen] (L) circle (6pt);
            \filldraw[nicegreen] (R) circle (6pt);
            \filldraw[nicegreen] (N) circle (6pt);
            \filldraw[nicegreen] (P) circle (6pt);
            \filldraw[nicegreen] (Q) circle (6pt);
            \filldraw[nicegreen,draw=orange1,line width=.5mm] (Q) circle (6pt) node[white] {\textbf{2}};
            \filldraw[niceblue,draw=orange1,line width=.5mm] (S) circle (6pt) node[white] {\textbf{1}};
            \draw[red,dashed,line width=1mm] ($(SL)!4mm!90:(S)$) -- ($(SL)!4mm!90:(L)$);
            \draw[red,dashed,line width=1mm] ($(KG)!4mm!90:(K)$) -- ($(KG)!4mm!90:(G)$);
			\end{tikzpicture}}%
		
		\subfloat[$31$-MinCut with $\RC\approx 0.016$.\\ Updated $\tau = (1,2,3,3)$]{
			\centering
			\begin{tikzpicture}[x=8mm,y=8mm,rotate=180]
			\coordinate (O) (0,0);
            \node[inner sep=0pt] (A) at (1.2,6.1) {A};
            \node[inner sep=0pt] (B) at (2.0,1.3) {B};
            \node[inner sep=0pt] (C) at (2.7,2.1) {C};
            \node[inner sep=0pt] (D) at (1.5,2.4) {D};
            \node[inner sep=0pt] (E) at (2.6,1.2) {E};
            \node[inner sep=0pt] (F) at (2.2,4.8) {F};
            \node[inner sep=0pt] (G) at (4.2,2.3) {G};
            \node[inner sep=0pt] (H) at (2.7,3.1) {H};
            \node[inner sep=0pt] (I) at (6.6,2.2) {I};
            \node[inner sep=0pt] (J) at (3.2,4.3) {J};
            \node[inner sep=0pt] (K) at (5.5,1.9) {K};
            \node[inner sep=0pt] (L) at (7.2,3.5) {L};
            \node[inner sep=0pt] (M) at (7.8,5.0) {M};
            \node[inner sep=0pt] (N) at (8.7,2.5) {N};
            \node[inner sep=0pt] (O) at (8.4,6.4) {O};
            \node[inner sep=0pt] (P) at (8.2,1.4) {P};
            \node[inner sep=0pt] (Q) at (7.7,2.2) {Q};
            \node[inner sep=0pt] (R) at (8.2,3.4) {R};
            \node[inner sep=0pt] (S) at (6.1,4.7) {S};
            \node[inner sep=0pt] (T) at (6.2,6.2) {T};
            \node[inner sep=0pt] (U) at (4.1,4.6) {U};
            \node[inner sep=0pt] (V) at (5.0,3.6) {V};
            \node[inner sep=0pt] (W) at (3.6,5.7) {W};
            \node[inner sep=0pt] (X) at (7.3,5.8) {X};

            \draw[black!50, line width=0.72mm] (A) -- (F) -- (J) -- (W) -- (U) -- (V) node[pos=0.5, inner sep=0pt, outer sep=0pt] (UV) {} -- (S) node[pos=0.5, inner sep=0pt, outer sep=0pt] (SV) {} -- (X) -- (T) -- (X) -- (O) -- (M) -- (X);
            \draw[black!50, line width=0.72mm] (L) -- (Q) -- (R) -- (N) -- (P) -- (Q) -- (I) -- (K) -- (G) node[pos=0.5, inner sep=0pt, outer sep=0pt] (KG) {} -- (C) -- (E) -- (B) -- (C) -- (D);
            \draw[black!50, line width=0.72mm] (B) -- (D) -- (H) -- (C);
            \draw[black!50, line width=0.72mm] (G) -- (V);
            \draw[black!50, line width=0.72mm] (L) -- (I) -- (P);
            \draw[black!50, line width=0.72mm] (S) -- (T);
            \draw[black!50, line width=0.72mm] (N) -- (Q);

            \draw[black!50, line width=0.9mm] (H) -- (D) -- (C) -- (H) -- (J) node[pos=0.5, inner sep=0pt, outer sep=0pt] (JH) {} -- (U) -- (W) -- (F) -- (J);
            \draw[black!50, line width=0.9mm] (S) -- (M) -- (X) -- (S) -- (V);
            \draw[black!50, line width=0.9mm] (I) -- (Q) -- (P) -- (N) -- (Q);
            
            \draw[black!50, line width=1.08mm] (S) -- (L) node[pos=0.5, inner sep=0pt, outer sep=0pt] (SL) {} -- (Q) -- (R) -- (L);
            \draw[black!50, line width=1.08mm] (C) -- (G);
            \draw[black!50, line width=1.08mm] (I) -- (K);
            \draw[black!50, line width=1.08mm] (C) -- (E) -- (B) -- (C);
            
            \foreach \i in {A,B,...,X}{
                \filldraw[niceblue] (\i) circle (6pt);
            }
            \filldraw[nicegreen] (A) circle (6pt);
            \filldraw[nicegreen] (F) circle (6pt);
            \filldraw[nicegreen] (W) circle (6pt);
            \filldraw[nicegreen] (U) circle (6pt);
            \filldraw[nicegreen] (V) circle (6pt);
            \filldraw[nicegreen] (G) circle (6pt);
            \filldraw[nicegreen] (H) circle (6pt);
            \filldraw[nicegreen] (D) circle (6pt);
            \filldraw[nicegreen] (B) circle (6pt);
            \filldraw[nicegreen] (E) circle (6pt);
            \filldraw[nicegreen] (J) circle (6pt);
            \filldraw[nicegreen,draw=orange1,line width=.5mm] (C) circle (6pt) node[white] {\textbf{3}};
            \filldraw[niceblue,draw=orange1,line width=.5mm] (S) circle (6pt) node[white] {\textbf{1}};
            \draw[red,dashed,line width=1mm] ($(SV)!4mm!90:(S)$) -- ($(SV)!4mm!90:(V)$);
            \draw[red,dashed,line width=1mm] ($(KG)!4mm!90:(K)$) -- ($(KG)!4mm!90:(G)$);
			\end{tikzpicture}}%
%
		\subfloat[$43$-MinCut with $\RC\approx 0.024$.\\ Final $\tau = (1,2,3,4)$]{
			\centering
			\begin{tikzpicture}[x=8mm,y=8mm,rotate=180]
			\coordinate (O) (0,0);
            \node[inner sep=0pt] (A) at (1.2,6.1) {A};
            \node[inner sep=0pt] (B) at (2.0,1.3) {B};
            \node[inner sep=0pt] (C) at (2.7,2.1) {C};
            \node[inner sep=0pt] (D) at (1.5,2.4) {D};
            \node[inner sep=0pt] (E) at (2.6,1.2) {E};
            \node[inner sep=0pt] (F) at (2.2,4.8) {F};
            \node[inner sep=0pt] (G) at (4.2,2.3) {G};
            \node[inner sep=0pt] (H) at (2.7,3.1) {H};
            \node[inner sep=0pt] (I) at (6.6,2.2) {I};
            \node[inner sep=0pt] (J) at (3.2,4.3) {J};
            \node[inner sep=0pt] (K) at (5.5,1.9) {K};
            \node[inner sep=0pt] (L) at (7.2,3.5) {L};
            \node[inner sep=0pt] (M) at (7.8,5.0) {M};
            \node[inner sep=0pt] (N) at (8.7,2.5) {N};
            \node[inner sep=0pt] (O) at (8.4,6.4) {O};
            \node[inner sep=0pt] (P) at (8.2,1.4) {P};
            \node[inner sep=0pt] (Q) at (7.7,2.2) {Q};
            \node[inner sep=0pt] (R) at (8.2,3.4) {R};
            \node[inner sep=0pt] (S) at (6.1,4.7) {S};
            \node[inner sep=0pt] (T) at (6.2,6.2) {T};
            \node[inner sep=0pt] (U) at (4.1,4.6) {U};
            \node[inner sep=0pt] (V) at (5.0,3.6) {V};
            \node[inner sep=0pt] (W) at (3.6,5.7) {W};
            \node[inner sep=0pt] (X) at (7.3,5.8) {X};

            \draw[black!50, line width=0.72mm] (A) -- (F) -- (J) -- (W) -- (U) -- (V) node[pos=0.5, inner sep=0pt, outer sep=0pt] (UV) {} -- (S) node[pos=0.5, inner sep=0pt, outer sep=0pt] (SV) {} -- (X) -- (T) -- (X) -- (O) -- (M) -- (X);
            \draw[black!50, line width=0.72mm] (L) -- (Q) -- (R) -- (N) -- (P) -- (Q) -- (I) -- (K) -- (G) node[pos=0.5, inner sep=0pt, outer sep=0pt] (KG) {} -- (C) -- (E) -- (B) -- (C) -- (D);
            \draw[black!50, line width=0.72mm] (B) -- (D) -- (H) -- (C);
            \draw[black!50, line width=0.72mm] (G) -- (V);
            \draw[black!50, line width=0.72mm] (L) -- (I) -- (P);
            \draw[black!50, line width=0.72mm] (S) -- (T);
            \draw[black!50, line width=0.72mm] (N) -- (Q);

            \draw[black!50, line width=0.9mm] (H) -- (D) -- (C) -- (H) -- (J) node[pos=0.5, inner sep=0pt, outer sep=0pt] (JH) {} -- (U) -- (W) -- (F) -- (J);
            \draw[black!50, line width=0.9mm] (S) -- (M) -- (X) -- (S) -- (V);
            \draw[black!50, line width=0.9mm] (I) -- (Q) -- (P) -- (N) -- (Q);
            
            \draw[black!50, line width=1.08mm] (S) -- (L) node[pos=0.5, inner sep=0pt, outer sep=0pt] (SL) {} -- (Q) -- (R) -- (L);
            \draw[black!50, line width=1.08mm] (C) -- (G);
            \draw[black!50, line width=1.08mm] (I) -- (K);
            \draw[black!50, line width=1.08mm] (C) -- (E) -- (B) -- (C);
            
            \foreach \i in {A,B,...,X}{
                \filldraw[niceblue] (\i) circle (6pt);
            }
            \filldraw[nicegreen] (F) circle (6pt);
            \filldraw[nicegreen] (A) circle (6pt);
            \filldraw[nicegreen] (W) circle (6pt);
            \filldraw[nicegreen] (U) circle (6pt);
            \filldraw[niceblue,draw=orange1,line width=.5mm] (C) circle (6pt) node[white] {\textbf{3}};
            \filldraw[nicegreen,draw=orange1,line width=.5mm] (J) circle (6pt) node[white] {\textbf{4}};
            \draw[red,dashed,line width=1mm] ($(UV)!4mm!90:(U)$) -- ($(UV)!4mm!90:(V)$);
            \draw[red,dashed,line width=1mm] ($(JH)!4mm!90:(J)$) -- ($(JH)!4mm!90:(H)$);
			\end{tikzpicture}}%
    \caption{Illustration of the \xistnameref for the Ratio Cut functional on a weighted toy graph, where edge thickness is proportional to edge weight. The vector $\tau$ in the \xistnameref determines the vertices $s$ and $t$ for the next $st$-MinCut. (a)~depicts the set of local maxima $\Vloc=\{1,2,3,4\}$. \xistref computes first the $21$-MinCut and updates $\tau$ in (b), then the $31$-MinCut with another update to $\tau$ in (c). Finally, the $43$-MinCut is computed, and $\tau$ is not updated further. Since the $31$-MinCut in (c) gives the best Ratio Cut value among all three partitions, \xistref outputs this partition and value.}
    \label{img:xist_explained}
\end{figure*}

\noindent We improve the \xvstnameref by two techniques, see \cref{img:xist_explained} for an illustration. 

First, we further restrict the minimization of the graph cut functional by only considering certain vertices $s$ and $t$ to compute the $st$-MinCuts over. In practice, $st$-MinCut partitions only become viable if one forces $s$ and $t$ to belong to tightly connected clusters to avoid a partition that cuts out only one vertex (this may happen as MinCut does not have a balancing term to counteract this). If $s$ and $t$ are connected to their respective neighbouring nodes through high-weighted edges, the $st$-MinCut cannot simply cut out one or the other and is therefore forced to find a different, more balanced way to separate both vertices. Formally, we call a vertex $u\in V$ a \emph{local maximum} if $\deg(u)\geq\deg(v)$ for all $v\in V$ with $\{u,v\}\in E$. We denote the set of local maxima as
\[
\Vloc\coloneqq \bigl\{u\in V\mid \deg(u)\geq \deg(v)\text{ for all } v\in V\text{ with }\{u,v\}\in E\bigr\},
\]
with its cardinality $\abs{\Vloc}=:N$. For a visualization of $\Vloc$ see \cref{img:xist_explained}~(a). By definition of $\Vloc$, we can ensure the scenario described above by requiring both $s$ and $t$ to be local maxima.

Second, it is not necessary to iterate over all pairs $\{s,t\}\subset\Vloc$ to obtain all $st$-MinCuts of vertices in $\Vloc$. It is known that in a graph of $n$ vertices, there are at most $n-1$ distinct $st$-MinCuts, and that these can be computed through the construction of the so-called \emph{Gomory-Hu tree} that was introduced by \cite{GomoryHu1961}. The Gomory-Hu tree is a tree built on $V$ where the edge weights are $st$-MinCut values, $s,t\in V$. \citet{GomoryHu1961} showed that this tree can be constructed through vertex contraction and only $n-1$ $st$-MinCut computations, and that it encapsulates all $st$-MinCut values. Consequently, there are only $n-1$ $st$-MinCuts, meaning that it is possible to improve the complexity of the \xvstnameref by $\BO(n)$. Additionally, their proofs can be adapted for the case that only those $st$-MinCuts are of interest where $s,t\in A$ for any subset $A\subseteq V$. We present this in \cref{apdx:sub:xist_consistency_proof}. 

Expanding upon this classical result, \citet{Gusfield1990} showed that alternatively to the Gomory-Hu method of vertex contraction and tree construction, it is possible to compute all $st$-MinCuts directly on the (uncontracted) graph $G$. Consequently, \citet[Section 3.4]{Gusfield1990} presented an adaptation of the Gomory-Hu algorithm that is simpler to implement and runs on the original graph only.

We incorporate the two techniques (restriction to local maxima and Gomory-Hu tree vertex selection) into the \xvstnameref to obtain the final \xistnameref (short for \textbf{X}C \textbf{i}mitation through $\mathbf{st}$-MinCuts; pronounced like \enquote{exist}).

\begin{algorithm}[!t]
\label{alg:xist}\label{alg:xvst_loc_max_merge}
\SetAlgorithmName{Xist algorithm}{Xist}
\SetAlgoLined
\DontPrintSemicolon
\SetKwInOut{Input}{input}
\SetKwInOut{Output}{output}
\SetKw{KwTerminate}{terminate}
\Input{weighted graph $G=(V,E,\mat{W})$}
\Output{XC value $c_{\min}$ with associated partition $S_{\min}$}
\BlankLine

set $c_{\min}\leftarrow\infty$ and $S_{\min}\leftarrow\emptyset$\;
determine the set of local maxima $\Vloc\subseteq V$\label{alg:xist:vloc}\;
\lIf{$N\coloneqq \abs{\Vloc} = 1$}{\KwTerminate \label{alg:xist:cs}}
set $\tau\leftarrow (1,\ldots,1)\in\R^N$\;
\For{$i\in\{2,\ldots,N\}$\label{alg:xist:forloop}}{
    let $s$ denote the $i$-th, and $t$ the $\tau_i$-th vertex in $\Vloc$\;
    compute an $st$-MinCut partition $S_{st}$ on $G$ \label{alg:xist:st-mincut}\;
    \tcp*{Note that by definition $s\in S_{st}$ and $t\in\comp{S}_{st}$}
    compute the XCut value $\XC_{S_{st}}(G)$ of partition $S_{st}$\label{alg:xist:xcut}\;
	\If{$\XC_{S_{st}}(G) < c_{\min}$}{\label{alg:xist:ifmin}
        $c_{\min} \leftarrow \XC_{S_{st}}(G)$\;
        $S_{\min} \leftarrow S_{st}$\;
	}\label{alg:xist:ifmin_end}
	\For{$j\in\{i,\ldots,N\}$}{
	    let $v_j$ be the $j$-th vertex in $\Vloc$\;
	    \lIf{$v_j\in S_{st}$ and $\tau_j = \tau_i$}{$\tau_j\leftarrow i$}
	}
}
\caption{XCut imitation through $st$-MinCuts on local maxima via implicit Gomory-Hu trees}
\end{algorithm}

Note that the set $\Vloc$ could be substituted by any subset of vertices $A\subseteq V$, and \xistref would still output the best XCut among $st$-MinCuts for all pairs of vertices $s,t\in A$, $s\neq t$. There are several reasons for choosing the set $\Vloc$ specifically:

\begin{enumerate}[label=(\roman*)]
    \item By only considering local maxima the $st$-MinCut is forced to separate $s$ and $t$ and therefore has to cut through the presumed \enquote{valley} (i.e.\ set of vertices with low degree) that lies between $s$ and $t$.
    \item Vertices that are no local maxima are not likely to benefit from $st$-MinCuts. This is due to the lack of a balancing term as previously discussed; MinCut tends to cut out only one vertex, e.g.\ the vertex of a low degree compared to its neighbours.
    This will, however, not happen as often with vertices of high degree as the latter punishes one-vertex cuts, by yielding a (comparatively) large cut value. See \cref{img:xist_explained}, for example.
    \item As a subset of $V$, the set $\Vloc$ greatly reduces the number of $st$-MinCuts to compute. The precise extend of its influence is difficult to quantify and heavily depends on the graph itself. In practice, only considering local maxima can lead to a vastly improved runtime (cf.\ \cref{img:xvst_comparison} later).
\end{enumerate}

The number $N$ of local maxima in a graph $G$ depends heavily on the vertex degrees and edge weights. Clearly, $N$ can be bounded from above by the independence number $\alpha(G)$ of the graph, several upper bounds of which are available in the literature \citep{William2011}. One example is the following upper bound \cite[Theorem~3.2]{SuilShiTaoqiu2021}:
\[
1\leq N\leq \alpha(G)\leq \frac{\Delta(G) n}{\Delta(G) + \delta(G)},
\]
where $\Delta(G)$ and $\delta(G)$ are the maximum and minimum number of neighbours of a vertex in $G$, respectively. This is not sharp in general, and is dominated by more sophisticated bounds, which, however, are more difficult to compute; in fact, the graph independence number $\alpha(G)$ is NP-hard to compute itself \citep{GareyJohnson1979}.

\subsection{Theoretical properties}
\label{sub:thp}

\noindent One important advantage of incorporating the Gomory-Hu method is that \xistref is guaranteed to optimize the cut value over $N-1$ \emph{distinct} partitions, thus removing redundant computations. This reduces runtime significantly (cf.\ \cref{img:xvst_comparison} later).

\begin{assumptions}
There exist $s,t\in V$, $s\neq t$, with a unique $st$-MinCut partition $S_{st}$ such that
\[
\XC_{S_{st}}(G)\; = \;\min_{u,v, S_{u,v}} \XC_{S_{uv}}(G),
\]
where the minimum is taken over all $u,v\in V$, $u\neq v$, and \emph{all} partitions $S_{uv}$ that attain the respective $uv$-MinCuts.
\label{assumptions:uniqueness}
\end{assumptions}

\cref{assumptions:uniqueness} is the technical condition necessary to guarantee that the \xvstnameref and in particular \xistref yield a consistent output regardless of the procedure chosen to compute the $st$-MinCut. The main issue is uniqueness of the underlying $st$-MinCuts: There could be two distinct partitions that both attain the $st$-MinCut, but yield a different XCut value. Even if we were not to rely on an oracle to compute the $st$-MinCut partition, problems could still arise as no polynomial algorithm can compute all $st$-MinCut partitions for fixed $s,t\in V$, $s\neq t$ \citep{Bonsma2010}, so it is not possible to efficiently determine the XCut minimizing among all $st$-MinCut attaining partitions, and it is often not clear what partition a given $st$-MinCut algorithm will output, given the existence of two partitions with the same MinCut value, but different XCut values.

Consequently, from a technical standpoint, a version of \cref{assumptions:uniqueness} is necessary for any algorithm that uses $st$-MinCut partitions and not just the cut value itself. The condition itself is fairly weak, especially in practice. For instance, if all $st$-MinCut partitions are unique (i.e.\ for any $s,t\in V$, $s\neq t$), \cref{assumptions:uniqueness} is satisfied, so global uniqueness would be a stronger restriction. In practice, this condition will almost always be satisfied, especially for image data such as the examples in \cref{sec:simulations}, or more generally for graphs with \enquote{suitably different} weights.

Using \cref{assumptions:uniqueness}, we can more clearly characterize the output of \xistref.

\begin{theorem}
    \xistref outputs $\min_{s,t\in\Vloc} \XC_{S_{st}}(G)$. Further, if \cref{assumptions:uniqueness} holds with $s,t\in\Vloc$, \xistref and the \xvstnameref yield the same output. If additionally the optimal XCut partition constitutes the st-MinCut for some $s,t\in\Vloc$, Xist outputs the optimal XCut partition.
\label{lem:xvst_loc_max_merge_equivalence}
\end{theorem}
\begin{proof}
    First, by design of \xistnameref, it computes the minimal XCut among \emph{some} $st$-MinCut partitions for \emph{all} pairs $s,t\in\Vloc$, $s\neq t$, the algorithm iterates over. This fact is shown in the \cref{apdx:sub:xist_consistency_proof}, specifically \cref{apdx:thm:xist_correctness}. This, however, already shows the first claim as \xistref considers $st$-MinCut partitions $S_{st}$ for all $s,t\in\Vloc$ and, by design, selects the one with the best XCut value among them.
    
    As for the second part of the statement, it is clear that the \xvstnameref computes $\min_{s,t\in V}\XC_{S_{st}}(G)$. Here, it is important to stress that the partition $S_{st}$ might not be the unique $st$-MinCut attaining partition. Indeed, as we treat each $st$-MinCut computation as an oracle call, it is not even clear whether in each algorithm, computing the $st$-MinCut consistently returns the same partition. Under \cref{assumptions:uniqueness}, however, there exist $s,t\in V$ (even $s,t\in\Vloc$ by assumption from the theorem statement) such that the $st$-MinCut partition $S_{st}$ is unique and it attains the best possible XCut among \emph{all possible} $uv$-MinCut partitions, for all $u,v\in V$. \cref{apdx:thm:xist_correctness} guarantees that this $st$-MinCut is considered by \xistref (and obviously also by the \xvstnameref) and, because the minimum $S_{st}$ is unique and attains the best XCut value, both algorithms yield $S_{st}$ as their output.

    The last assertion of the theorem now follows immediately.
\end{proof}

\cref{lem:xvst_loc_max_merge_equivalence} shows that the \xvstnameref and \xistref are equivalent up to restriction to the subset $\Vloc$, and in particular that it suffices to consider $N-1$ pairs of $st$-MinCut partitions to obtain all possible such cuts between pairs $s,t\in\Vloc$, $s\neq t$. In particular, this improves the worst-case complexity of \xistref (compared to the \xvstnameref) by one order of magnitude. We further obtain an approximation guarantee, i.e.\ that under the assumption that the optimal XCut partition (i.e.\ the partition attaining the minimum $\min_{S\subset V} \XC(S)$) is an $st$-MinCut for some $s,t\in\Vloc$, then \xistref outputs this partition.

The following \cref{lem:xist_complexity} shows that use of both the restriction to local maxima and the use of the Gomory-Hu method improves the runtime of \xistref significantly when compared to the \xvstnameref.

\begin{theorem}
Assume that for any fixed partition $S$, the evaluation of \;$\XC_S({G})$ takes $\BO(\kappa)$ time, with $\kappa \coloneqq \kappa(m,n)$ depending on $m$ and~$n$. Then, the computational complexity of the \xvstnameref is $\BO(n^2 \max\{nm,\kappa\})$, and that of \xistref is $\BO(N \max\{n m,\kappa\})$.
\label{lem:xist_complexity}
\end{theorem}

Note that evaluation of the functionals of popular cuts such as MinCut, Ratio Cut, NCut or Cheeger Cut are computed using only edge weights, meaning that this step is actually only $\BO(m)$, except for Ratio Cut, which also requires determining $\abs{S_{st}}$, thus yielding $\BO(m+n)$ instead. Thus, for such cuts, the \xvstnameref admits a computational complexity of $\BO(n^3m)$, while \xistref is $\BO(Nnm)$ for general graphs. Recall that spectral clustering has a worst-case complexity of $\BO(n^3)$. Thus, if $Nm=\BO(n^2)$, the \xistnameref is at least as fast as spectral clustering. If further $m = \BO(n)$ and $N \ll n$, \xistref can be much faster; this is often the case for graphs of bounded degrees, e.g.\ for an image, where its pixels constitute a regular grid. Note, however, that in the least favorable case of $N \asymp n$ and $m \asymp n^2$, \xistref can be one order slower than spectral clustering.

\subsection{Software implementation}

We have implemented the \xvstnameref, the \xistnameref and \multixistref in Python (also in R, relatively slower). Our implementations can be found on GitHub \citep{Xist_Github}, to allow for the results in the following \cref{sec:simulations} to be reproduced. We are currently in the process of implementing our algorithms in C++, from which we expect a significant increase in computational efficiency.

\section{Simulations and applications}
\label{sec:simulations}

The following simulations were performed on a laptop with Windows 11 operating system, a 2.70 GHz Intel\textsuperscript{\textregistered} Core\textsuperscript{\texttrademark} i5-12600H processor and 16 GB of RAM. The code necessary to reproduce the following (data) analysis can be found on GitHub \citep{Xist_Github}.

\subsection{Approximation of multiway cuts}
\label{sub:extension_to_multiway_cuts}

\noindent We start with an extension of \xistref to compute multiway cuts, and examine its performance on a real-world dataset. Given a graph $G$ and a number of desired partitions $k$, this extension \multixistref is obtained through a greedy approach: First, apply \xistref to $G$, receiving partitions $T_{1,1}$ and $T_{1,2}:=V\setminus T_{1,1}$, then restrict $G$ to $T_{1,1}$ (and $T_{1,2}$) to obtain $G_{1,1}$ (and $G_{1,2}$). Then cut both restricted graphs again using \xistref, and select the partition ($T_{1,1}$, say) that yields the lowest (\enquote{normalized}; see below) \xistref value. Defining $T_{2,1}$ and $T_{2,2}:=T_{1,1}\setminus T_{2,1}$, at this point $G$ is considered to be divided into three partitions: $T_{1,2}$, $T_{2,1}$ and $T_{2,2}$. The pattern continues: The graph is further restricted to $T_{2,1}$ and $T_{2,2}$, and again cut using \xistref, whereupon the lowest cut value among those and $T_{1,2}$ is selected. This iterative cutting, restricting and selecting is continued until $k$ partitions have been computed.

It should be noted that it is necessary to \enquote{normalize} the computed \xistref cut values in order to ensure comparability across graphs of (possibly) vastly different sizes (in terms of $n$, $m$ and $\mat{W}$). 
More specifically, cutting a graph whose weights have been scaled up by a constant factor should not change its normalized XCut. Hence, the balanced graph cuts are normalized by multiplying with an additional factor of $\sum_{i,j} w_{ij}/\bal(V,V)$:
\begin{align*}
\overline{\XC}_S(G) &:= \sum_{i\in S, j\in\comp{S}} \frac{w_{ij}}{\bal(S,\comp{S})} \frac{\bal(V,V)}{\sum_{i,j\in V} w_{ij}},\\
\quad\text{s.t.}\quad \overline{\XC}(G) &:= \min_{S\subset V} \overline{\XC}_S(G).
\end{align*}
Clearly, a partition $S$ minimizes $\overline{\XC}(G)$ if and only if it minimizes $\XC(G)$ (regardless of whether this minimization is done over all partitions or only over a subset), so this normalization does not impact the output of \xistref (up to the aforementioned normalizing factor of $\sum_{i,j} w_{ij}/\bal(V,V)$).

\begin{algorithm}[!t]
\label{alg:xist_iterated}
\SetAlgorithmName{Multi-Xist algorithm}{Multi-Xist}
\SetAlgoLined
\DontPrintSemicolon
\SetKwInOut{Input}{Input}
\SetKwInOut{Output}{Output}
\Input{weighted graph $G=(V,E,\mat{W})$, number $k\in\N_{\geq 2}$ of desired partitions}
\Output{$k$-way partition $R$}
\BlankLine

$S\leftarrow \bigl((\emptyset,V),(\emptyset,\emptyset),\ldots,(\emptyset,\emptyset)\bigr)\in \bigl\{(T_1,T_2)\mid T_1,T_2\subseteq V\bigr\}^k$\;
$c\leftarrow (0,\infty,\ldots,\infty)\;\in\;(-\infty, \infty]^k$\;

\For{$i = 2,\ldots,k$}{
    set $j_{\min}\leftarrow \argmin_{\ell=1,\ldots,k} c_j$ and $(T_{i,1},T_{i,2}) \leftarrow S_{j_{\min}}$\;
    \For{$\ell=1,2$}{
    	\eIf{$\abs{T_{i,\ell}} \leq 1$}{
            set $r_{i,\ell}\leftarrow \infty$ and $R_{i,\ell}\leftarrow \emptyset$\;
    	}{
    	    set $G_{i,\ell}\leftarrow (T_{i,\ell}, E\cap (T_{i,\ell}\times T_{i,\ell}), \mat{W}|_{T_{i,\ell}})$, where $\mat{W}|_{T_{i,\ell}}$ is the weight matrix of $G$ restricted to the vertices in $T_{i,\ell}$ only\;
    	    compute $r_{i,\ell} \leftarrow \xist(G_{i,\ell})$ with corresponding partition $R_{i,\ell}$ using \xistref\label{step:multiway_xvst_cut}\;
    	    $r'_{i,\ell} \leftarrow r_{i,\ell}\cdot \bal_{G_{i,\ell}}(T_{i,\ell}, T_{i,\ell}) / \bigl(\sum_{i,j\in V} w_{ij} \bigr)$
    	}
    }
	set $c_{j_{\min}} \leftarrow r'_{i,1}$ and $c_i \leftarrow r'_{i,2}$\;
	set $S_{j_{\min}} \leftarrow (R_{i,1}, T_{i,1}\setminus R_{i,1})$ and $S_i \leftarrow (R_{i,2}, T_{i,2}\setminus R_{i,2})$\;
}
$R \leftarrow \{S_{i,1}\cup S_{i,2}\}_{i=1,\ldots,k}$\; 

\caption{Multiway Xist with XCut-minimizing cluster selection}
\end{algorithm}

One could also modify the Xist algorithm to approximate multiway graph cut $\kXC$ directly by computing $\{s_1,\ldots,s_k\}$-MinCuts for given nodes $s_1,\ldots,s_k\in V$. This problem is more complex than computing $\kMC(G)$ of the entire graph $G$, even being NP-complete for general graphs, although polynomial algorithms exist if $G$ is planar and $k$ is fixed, see \cite{Dahlhaus_etal1994}. In contrast, our choice of iterative cutting has the advantage that it is computable in polynomial time for any weighted graph, more precisely, in $\BO(kNnm)$ time. Moreover, \multixistref has a built-in \enquote{optimality guarantee} in that in each iteration the best current XC imitation is selected. This, of course, comes with all advantages and disadvantages that such a greedy approach entails.

We now demonstrate the application of \multixistref to a real-world example of detecting cell clusters. The data in question consists of images of microtubules in PFA-fixed NIH 3T3 mouse embryonic fibroblasts (DSMZ: ACC59) labeled with a mouse anti-alpha-tubulin monoclonal IgG1 antibody (Thermofisher A11126, primary antibody) and visualized by a blue-fluorescent Alexa Fluor\textregistered\ 405 goat anti-mouse IgG antibody (Thermofisher A-31553, secondary antibody). Acquisition of the images was performed using a confocal microscope (Olympus IX81). The images were kindly provided by Ulrike Rölleke and Sarah Köster (University of Göttingen), and are availiable on GitHub \citep{NIH3T3_dataset}. We take on the task of identifying the $k=8$ main clusters of cells. 
To reduce computation time, we construct the graph by down-sampling the original cell image on a coarse regular grid of size $r\times r$, here for $r=128$. While the partition computed using this slight discretization does not have the same resolution of the original image (which is $512\times 512$), the grid size parameter $r$ can be chosen as large or as small as required. Edges were assigned by connecting each grid point to its eight direct neighbours, with weights defined as the product of the grey-color intensity values of connected vertices. In contrast to the usual \xistnameref, instead of line~\ref{alg:xist:vloc} we defined the set of local maxima to be pixels whose image (grey) value (instead of their degree) is larger than that of all neighbouring pixels. We make this slight modification because in images, this method better encompasses and expresses the concept of \emph{local maxima}. The results are displayed in \cref{img:kncut_cell_cluster_example}.

\begin{figure*}[htpb]
    \centering
    \subfloat[\multixistref]{\includegraphics[width=0.47\textwidth]{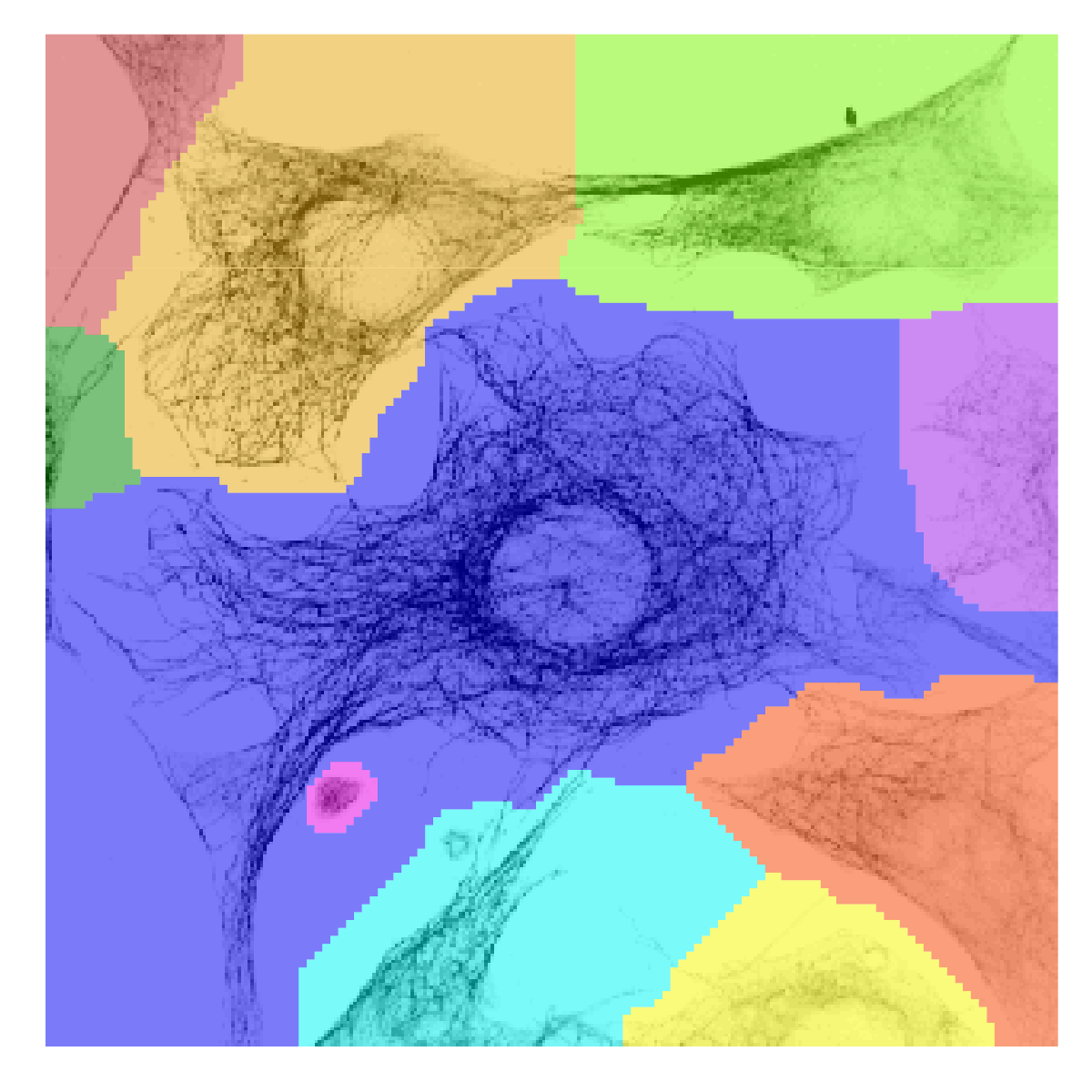}}%
    \hspace{2mm}
    \subfloat[Spectral clustering]{\includegraphics[width=.47\textwidth]{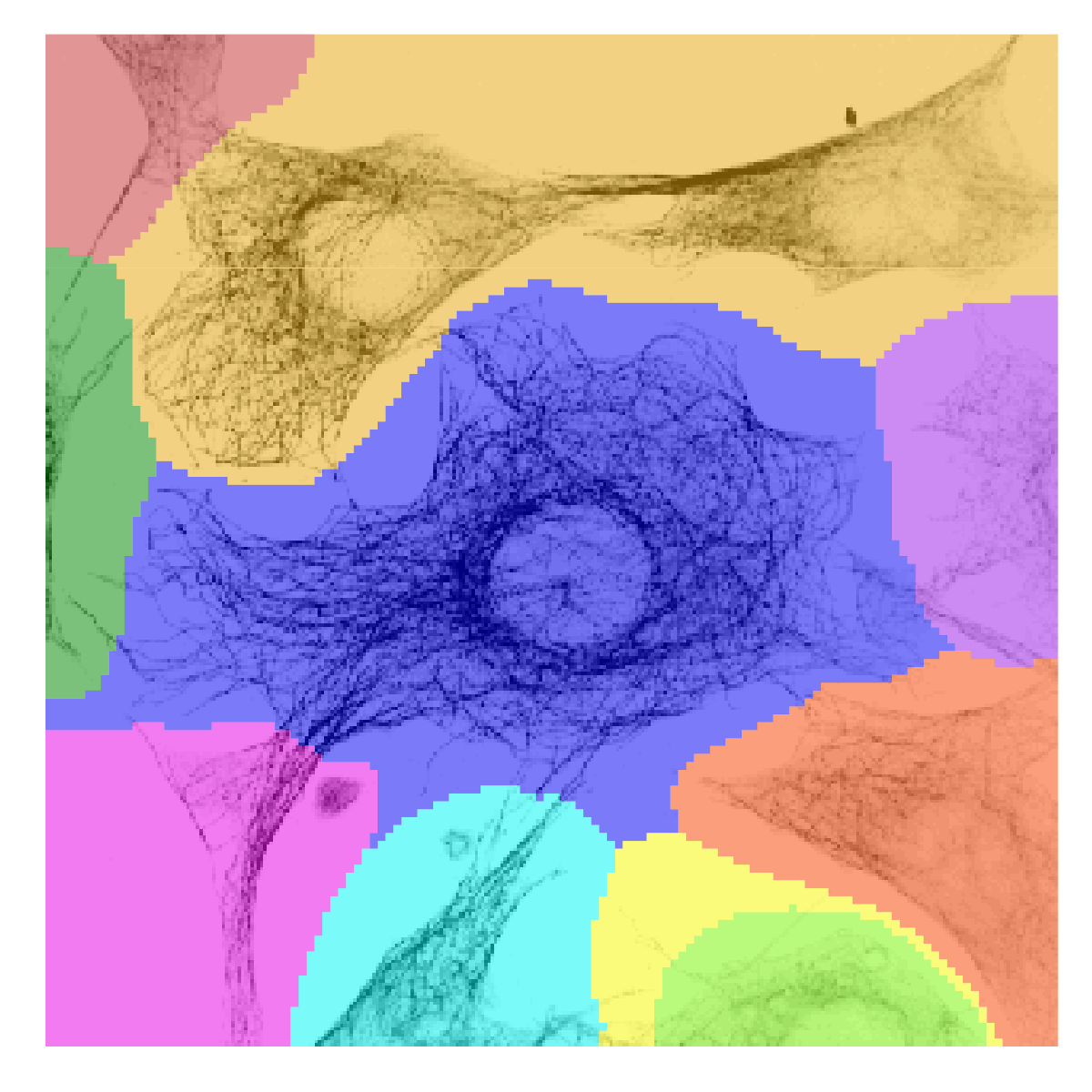}}
    \caption{Image segmentation of microtubules in NIH 3T3 cells. The colors visualize the resulting $k=10$ partitions via \multixistref and spectral clustering, respectively. The underlying cell clusters are visible in black, where a darker color marks denser microtubule network.}
    \label{img:kncut_cell_cluster_example}
\end{figure*}

We see that \multixistref does indeed yield a sensible partition, and it is also able to separate neighbouring cell clusters of different scales, correctly separating smaller, more disconnected cell clusters from the rest. This is in contrast to spectral clustering which cuts right through the main cell cluster, separating it into a red and blue part. This example exemplifies also the scalability issues spectral clustering faces when dealing with clusters of different scales \cite{NadlerGalun2007}. A visualization of the full \multixistref process on this image, in particular the selection of which cluster to cut further, is shown in \cref{img:kncut_cell_cluster_example}, where we let $k=2,\ldots,9$. This example demonstrates even more clearly the multiscale advantage \xistref has over spectral clustering can be found in \cref{apdx:img:kncut_vs_spec_clust} in \cref{apdx:sub:kncut_vs_spec_clust}.

\subsection{Empirical runtime comparison}
\label{sub:empirical_runtime_comparison}

\noindent In the following sections we compare \xistref to five state-of-the-art graph partitioning algorithms, both in terms of partition quality and runtime. The algorithms and implementation being compared are the following:

\begin{itemize}
    \item Our \xistref algorithm as implemented in {Python} (see \cite{Xist_Github}).
    \item The Leiden algorithm introduced by \cite{TraagWaltmanEck2019} and implemented in the \texttt{leidenalg} {Python} package \citep{leidenalg}.
    \item The {Python}-support of the KaHIP (Karlsruhe High Quality Partitioning) algorithm package \citep{KaHIP}, where we used the \enquote{strong social} mode of their main algorithm to achieve the highest partitioning quality.
    \item The main algorithm of the METIS graph partitioning software \citep{KarypisKumar1998} as wrapped in {Python} via the \texttt{PyMetis} package \citep{pymetis}.
    \item The Chaco algorithm \citep{HendricksonLeland1995} as a standalone software package obtained from GitHub \citep{chaco_github}, called through a shell script executed from within {Python} in order to use and record its output.
    \item The classical spectral clustering approach originally proposed by \cite{ShiMalik2000} and realized through the \texttt{scikit-learn} {Python} package \citep{scikit-learn}, where the eigenvector entries are clustered again using $k$-means clustering, see \cite[Section~4]{vonLuxburg2007}.
\end{itemize}

All of the above algorithms were used with their default parameters to ensure comparability, with the exception of the Leiden algorithm because it requires the additional input of a \enquote{resolution parameter} with no given default option. Hence, we considered the Leiden algorithm as an oracle (denoted as \enquote{Leiden (oracle)}), meaning so that the tuning parameter is chosen as to optimize the quality measure in question, i.e.\ to minimize the NCut value or maximize the classification rate (see \cref{img:4algs_classifrate+ncut}). Such choices of parameters require the knowledge of the true cluster assignments, thus referred to as \emph{oracle parameter} choices. Only when the cut quality is not evaluated, i.e.\ for runtime comparisons (see \cref{img:xvst_comparison}), we use the Leiden algorithm in its usual form (denoted as \enquote{Leiden}) with an arbitrary choice of resolution parameter.

Further note that the Leiden algorithm returns a partition of the graph into $k$ sets of vertices, where $k$ is heavily dependent upon the \enquote{resolution parameter}, so that, to evaluate the quality of the Leiden algorithm fairly, it is necessary to generalize the XCut term $\XC$ to a $k$-fold partition of $G$. This generalization is well-established in the literature, and it is given by
\begin{equation}
\XC_T(G) := \frac{1}{2}\sum_{i=1}^k \XC_{T_i}(G)\quad\text{for}\quad T = \{T_1,\ldots,T_k\}\subset V^k\text{ with }\sum_{i=1}^k T_i = V \label{eq:multiway_xcut}
\end{equation}
(and all $T_i$ are pairwise disjoint). Compared to the original definition of $\XC_S(G)$, where $S\subset V$, it is evident that $\XC_S(G) = \XC_{\{S,\comp{S}\}}(G)$.

We compare the runtime of the basic algorithm and the proposed $\xist$ algorithm with that of the above algorithms on the same cell image dataset as in \cref{sub:extension_to_multiway_cuts}. We crop the images to $504\times 504$ and vary the \enquote{resolution} of the grid (i.e.\ the grid size parameter $r$) and examine the rates at which the runtime increases, see \cref{img:xvst_comparison}. Note that we did not include the Chaco algorithm in this runtime comparison as it does not have a Python implementation (or a Python wrapper), so it had to be run as a standalone program, rendering time comparisons meaningless. To still give an intuition in terms of absolute time, Chaco's runtime was slightly above the times of KaHIP in \cref{img:xvst_comparison}.

\begin{figure}[!t]
    \centering
    \includegraphics[width=\linewidth]{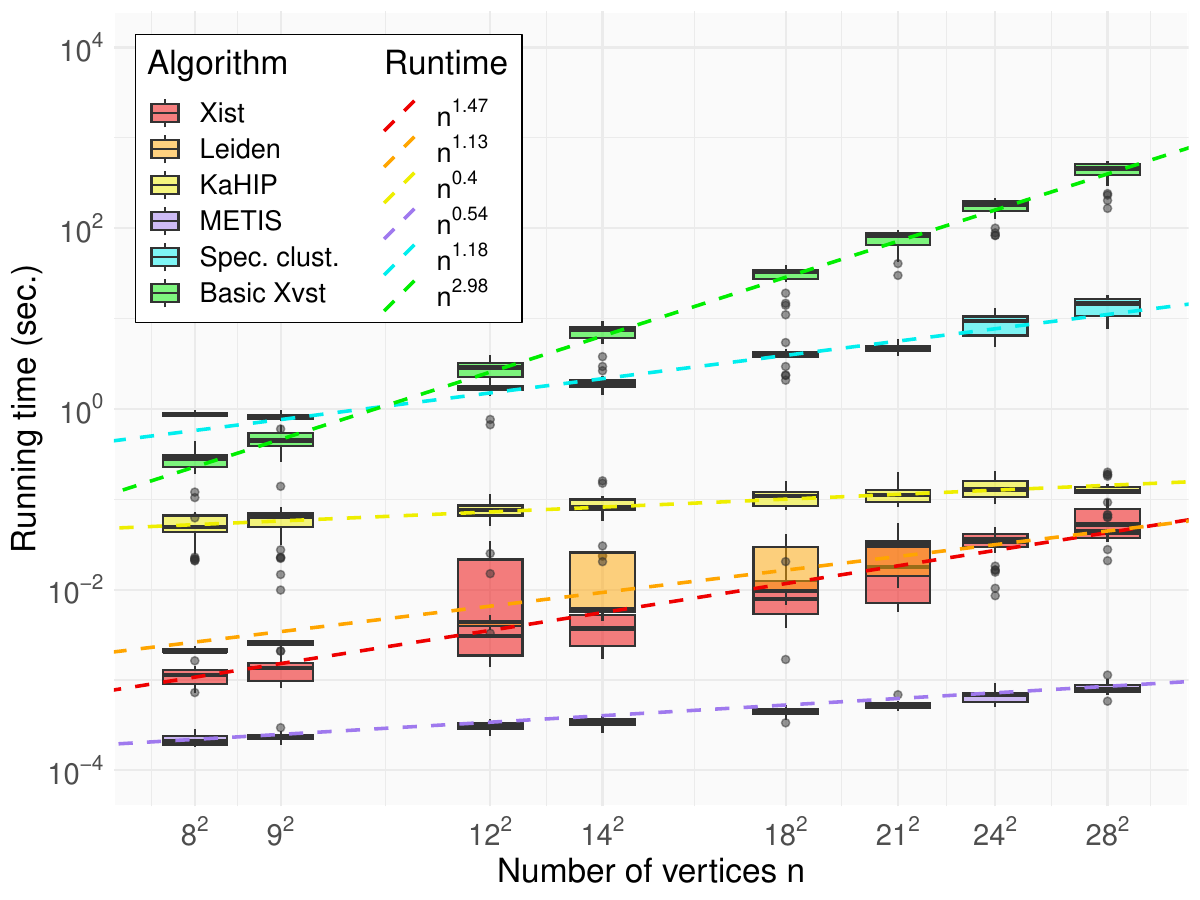}
    \caption{Comparison of the empirical runtimes (in seconds) of \xistref (red), Leiden in its non-oracle form (orange), KaHIP (yellow), METIS (violet), spectral clustering (cyan), and the \xvstnameref (green) on a $\log$-$\log$ scale against the number $n=r^2$ of vertices, with dashed lines indicating the empirical complexity of the respective algorithms. The boxplots were obtained by applying each algorithm to a dataset of 21 NIH3T3 cell cluster images of size $504\times 504$ pixels each after discretizing it onto a $r\times r$ grid, for $r\in\{8,9,12,14,18,21,24,28\}$.}
    \label{img:xvst_comparison}
\end{figure}

As is shown, the proposed \xistnameref is empirically roughly 1.5 orders of magnitude faster than the \xvstnameref. Moreover, in terms of absolute time it is much faster than even very efficient algorithms such as KaHIP or spectral clustering which are implemented in {C}, even though \xistref is only implemented in {Python} and the fact that for the computation of the $st$-MinCuts, a theoretically suboptimal algorithm had to be used as -- to the best of our knowledge -- there does not yet exist any implementation of the current fastest ones, namely \cite{Orlin2013} and \cite{OrlinGong2021}. We are currently in the process of implementing \xistref efficiently (i.e.\ in {C++}). Finally, note that the near-constant time of the Chaco algorithm in \cref{img:xvst_comparison} should be considered with an appropriate amount of scepticism since it is the only algorithm without a {Python} implementation and thus needs to be called via a shell script as described above.

\subsection{Qualitative assessment of \texorpdfstring{\xistref}{Xist}}
\label{sub:xvstlocmax_quality}

\noindent Since the aim of our algorithm is to imitate graph cuts (while still being computable in polynomial time), we evaluate the quality of our algorithm in a partitioning exercise where the task is to determine which points from a two-cluster sample stem from which cluster. To this end let $\delta>0$ and consider a random mixture of Gaussians:
\begin{equation}
\vec{X}\coloneqq (X_1,\ldots,X_n)\sim B Z_0 + (1-B) Z_{\delta},
\label{eq:gaussian_mixture}
\end{equation}
where $B\sim\mathrm{Ber}({1}/{2})$ a Bernoulli random variable, and $Z_{\delta}\sim \mathcal{N}_2(\vec{\delta}, \mat{I}_2)$ the two-dimensional standard Gaussian around $\vec{\delta}=(\delta,\delta)\in\R^2$. For details on the graph construction see \cref{img:4algs_classifrate+ncut}.

\begin{figure*}[!t]
    \centering
    \subfloat[NCut value]{\includegraphics[width=.47\textwidth]{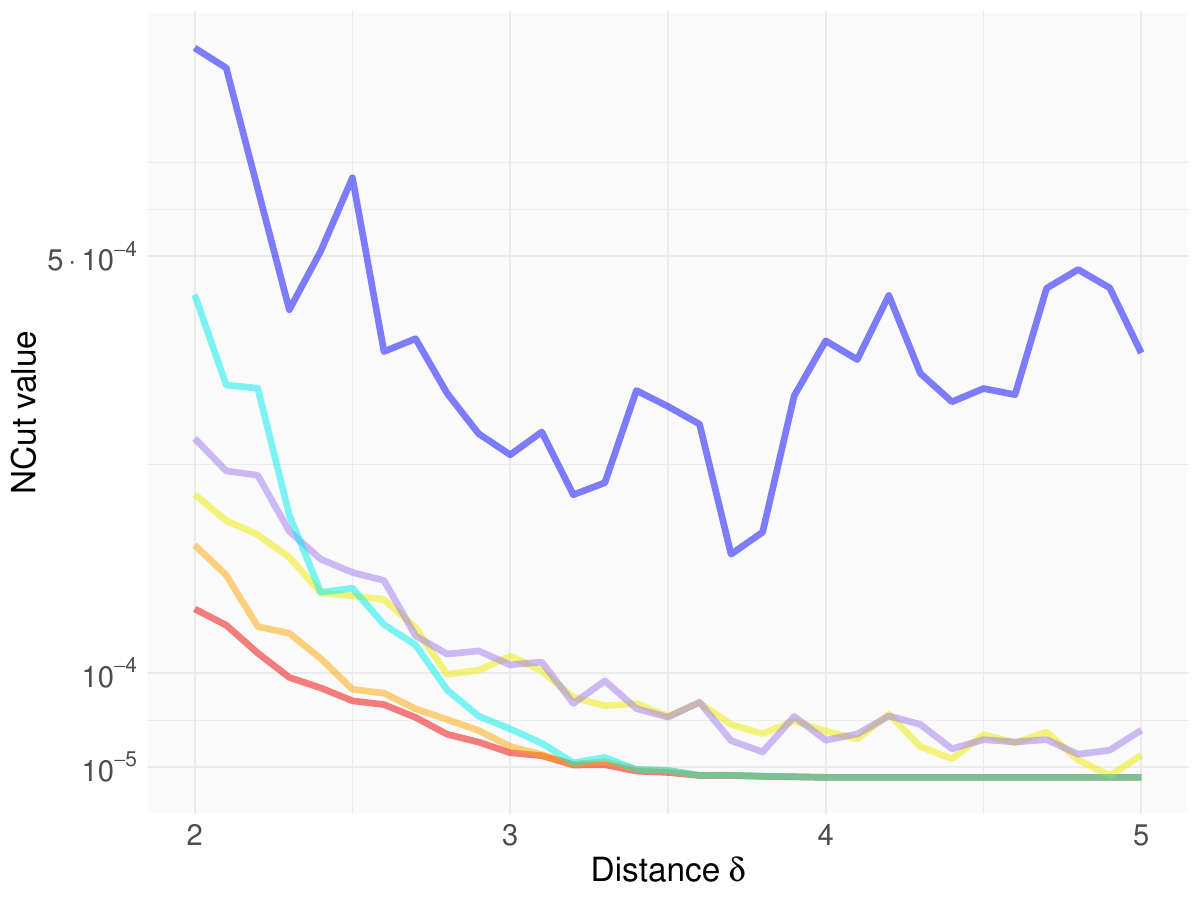}}
    \hspace{5mm}
    \subfloat[Classification rate]{\includegraphics[width=.47\textwidth]{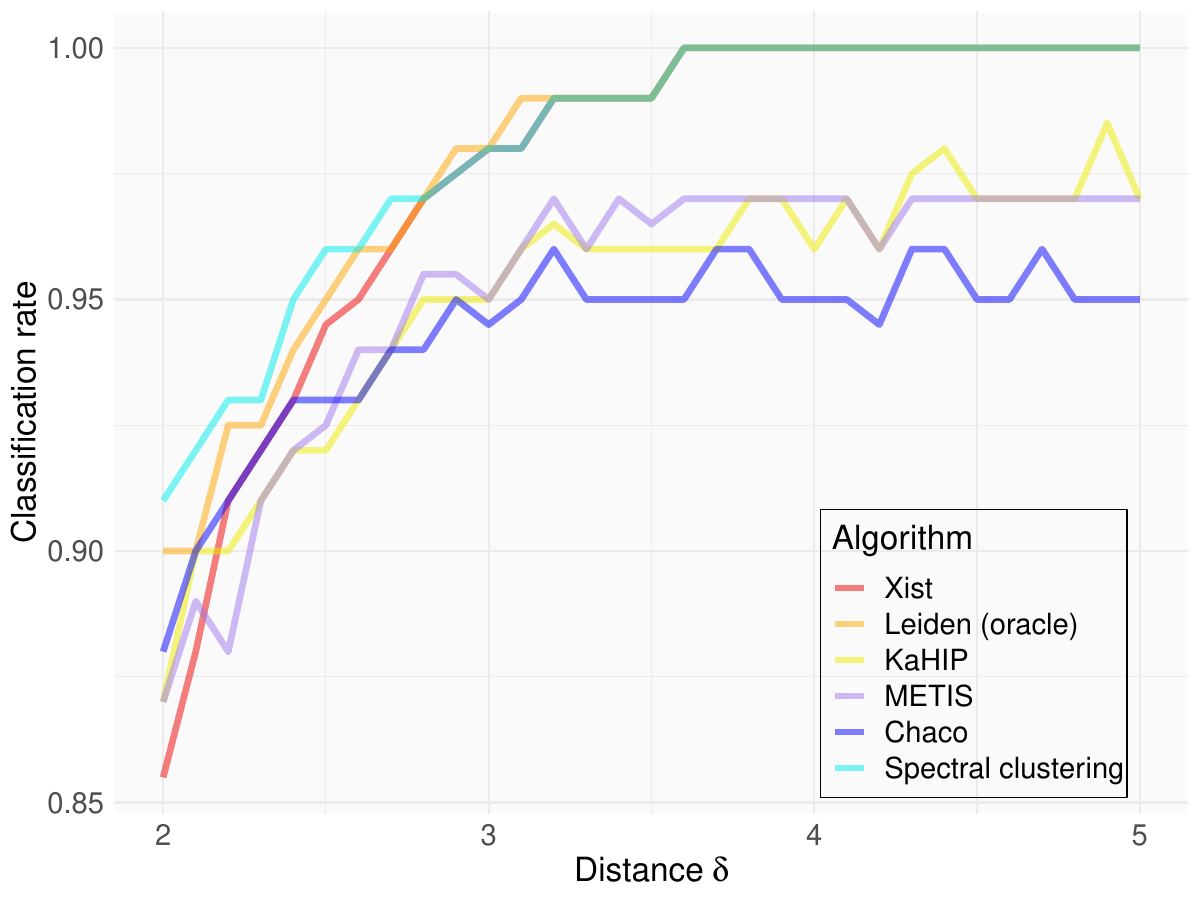}}
    \caption{Comparison of the classification rate and NCut value of \xistref (red), the oracle version of the Leiden algorithm (orange), KaHIP (yellow), METIS (violet), Chaco (blue), spectral clustering (cyan), and the \xvstnameref (green) over the intercluster distance $\delta$. The graph being partitioned is built from $n=100$ samples from the Gaussian mixture distribution in \eqref{eq:gaussian_mixture} using a combination of 5-nearest and 0.2-neighbourhood, i.e.\ by defining $\{u,v\}\in E$ if and only if $ u\in\mathrm{NN}_5(v)$ or $v\in\mathrm{NN}_5(u)$ or $\norm{u-v}\leq 0.2$, where $\mathrm{NN}_5(u)$ denotes the set of five nearest (in terms of the euclidean norm $\norm{\cdot}$ on $\R^2$) neighbours in $V$ of $u\in V$. The weights are defined as $w_{uv}\coloneqq \exp(-\norm{u-v}/0.2)$. These choices are made to ensure a connected graph (5-nearest neighbours) where clusters are more easily detectable (0.2-neighbourhood) and the spatial structure is pronounced (exponential weights). The curves depicted where obtained by taking the mean classification rate and mean NCut value over $100$ iterations of the above procedure, re-generating $\vec{X}$ each time.}
    \label{img:4algs_classifrate+ncut}
\end{figure*}

We compare the partition quality of \xistref to that of the other algorithms (see \cref{sub:empirical_runtime_comparison}). Note that \citet{LoefflerZhangZhou2021} have shown that in this scenario (i.e.\ random mixture of Gaussians), spectral clustering is asymptotically minimax optimal in terms of the classification rate (i.e.\ the ratio of observations $X_i$ correctly classified as belonging to their respective Gaussian in the mixture), so this scenario is highly favourable towards spectral clustering. Despite this, \xistref is able to achieve a similar accuracy to the oracle version of the Leiden algorithm, being only slightly worse than spectral clustering and much better than the other state-of-the-art algorithms as \cref{img:4algs_classifrate+ncut}~(a) shows, where the classification rate is plotted over the intercluster distance $\delta$. Note that, as stated before, Leiden was realized as an oracle, i.e.\ it outputs the highest classification rate in (a) and lowest NCut value in (b) over a range of its so-called \enquote{resolution parameter}. Likely due to this realization as an oracle it performs almost as good as spectral clustering, but only marginally better than \xistref, and only for small $\delta$. \xistref consistently outperforms the other state-of-the-art algorithms. The difference is most striking in \cref{img:4algs_classifrate+ncut}~(b), where \xistref yields a consistently better NCut value than all the other algorithms, thus achieving its set goal of approximating graph cuts better than spectral clustering in particular.

This simulation study additionally refutes the notion that graph-based methods (such as KaHIP, METIS or Chaco) can not achieve a good performance in the case of a Gaussian mixture \emph{by design} because they introduce an additional abstraction (the graph). The success of \xistref in this scenario shows that the poor performance of other graph-based methods is not inherent to the graph structure, but likely the result of computational artifacts. Indeed, as \cref{img:4algs_classifrate+ncut} demonstrates, if one is able to compute graph cuts highly accurately, as good as or even better results than spectral clustering can be achieved.

\subsection{Clustering large network datasets}

\noindent To demonstrate the applicability and versatility of \xistref, even in its current {Python} implementation, we take on the task of clustering large network datasets. We consider the \enquote{Stanford Large Dataset Collection} (SNAP, see \cite{SNAP}) which contains many real-world network graphs of varying types and sizes. Of those graphs we consider several undirected graphs, all unweighted, in particular
\begin{itemize}
    \item the arXiv High-Energy Physics -- Phenomenology (HEP-PH) dataset \citep{CA/CIT-HEPH1,CA/CIT-HEPH2} where authors represent nodes and edges papers co-authored between them;
    \item parts of the Multi-Scale Attributed Node Embedding (MUSAE) dataset \citep{MUSAE}, namely
    \begin{itemize}
        \item the MUSAE Facebook subset, where nodes are verified Facebook pages and edges mutual links between them;
        \item the MUSAE Squirrel subset, where nodes are squirrel-related Wiki-pedia pages and edges mutual links between them;
    \end{itemize}
    \item the Artist subset of the GEMSEC Facebook dataset \citep{GEMSEC}, consisting of nodes representing verified Facebook pages categorized as \enquote{artist} and edges mutual likes between them;
    \item and the Enron email dataset \citep{EMAIL-ENRON2}, containing Enron email addresses as nodes and edges between nodes if an email was sent from one node to the other;
\end{itemize}
In all datasets, node features were ignored as the algorithms in question are not designed to account for such information. Also note that the selection of subsets of these datasets was essentially arbitrary; similar results hold for other subsets (e.g.\ the MUSAE crocodile subset).

On the above datasets, we compare the performance of \xistref and the KaHIP, METIS and Chaco algorithms, both in terms of the Normalized Cut value of the partition as well as the algorithm runtime. The result can be found in \cref{tbl:xist_vs_leiden_large_datasets}. Note that spectral clustering cannot be reasonably applied here since it requires computation of the eigenvalues of the graph Laplacian, an $n\times n$ matrix, a task that is very memory-intensive. We also did not apply the Leiden algorithm here since the search for an NCut-minimizing \enquote{resolution parameter} would be very time-consuming for datasets as large as the above, and even then reporting the time accurately would be difficult since the runtime (and output) of any such oracle heavily depends on the range the \enquote{resolution parameter} is chosen from.

\begin{table}[htpb]
    \centering
    \caption{Normalized Cut value and empirical runtime of \xistref as well as KaHIP, METIS and Chaco for various datasets. Here, $n$ and $m$ are the number of vertices and edges of only the largest connected component of the respective dataset.}
    \label{tbl:xist_vs_leiden_large_datasets}
    \begin{tabular}{crrrrrr}
    \toprule
    \multicolumn{3}{c}{} & \multicolumn{4}{c}{NCut value $(\cdot 10^{-8})$} \\
    \cmidrule(lr){4-7}
        Dataset & $n$ & $m$ & $\xist$ & KaHIP & METIS & Chaco  \\
        \midrule
        MUSAE Squirrel & 5201 & 198353 & $\mathbf{26.08}$ & $45.76$ & $67.33$ & $43.98$ \\
        arXiv HepPh & 11204 & 117619 & $\mathbf{1.01}$ & $45.73$ & $51.13$ & $55.74$ \\
        MUSAE Facebook & 22470 & 170823 & $\mathbf{3.01}$ & $17.23$ & $18.66$ & $20.91$ \\
        Enron Email & 33696 & 180811 & $\mathbf{1.25}$ & $67.70$ & $53.06$ & $66.91$ \\
        GEMSEC Artist & 50515 & 819090 & $\mathbf{4.70}$ & $13.56$ & $13.87$ & $14.98$ \\
    \bottomrule
    \end{tabular}%
    \vspace{4mm}
    \begin{tabular}{crrrrrr}
    \toprule
    \multicolumn{3}{c}{} & \multicolumn{4}{c}{Time (seconds)} \\
    \cmidrule(lr){4-7}
        Dataset & $n$ & $m$ & $\xist$ & KaHIP & METIS & Chaco \\
        \midrule
        MUSAE Squirrel & 5201 & 198353 & $0.38$ & $1.26$ & $\mathbf{0.04}$ & $0.09$ \\
        arXiv HepPh & 11204 & 117619 & $0.91$ & $0.76$ & $\mathbf{0.03}$ & $0.04$ \\
        MUSAE Facebook & 22470 & 170823 & $8.95$ & $1.10$ & $\mathbf{0.05}$ & $0.07$ \\
        Enron Email & 33696 & 180811 & $4.67$ & $3.86$ & $\mathbf{0.06}$ & $0.36$ \\
        GEMSEC Artist & 50515 & 819090 & $10.71$ & $20.44$ & $\mathbf{0.27}$ & $0.32$ \\
    \bottomrule
    \end{tabular}
\end{table}

Note that we restricted the original datasets to only their largest respective components since, as previously mentioned, \xistref would immediately separate unconnected components and return an XCut value of $0$ (for any XCut balancing term). It is evident from \cref{tbl:xist_vs_leiden_large_datasets} that \xistref is able to handle large datasets, taking a moderate amount of time to cluster all the datasets without being implemented in a lower-level programming language such as {C} or {C++} -- unlike METIS, KaHIP or Chaco which have efficient implementations already (see \cref{sub:empirical_runtime_comparison}). Despite this, \xistref has a better runtime scaling than KaHIP, and we expect a {C++} implementation of it to be competitive in terms of runtime.

Qualitatively, one can see that \xistref decisively outperforms all three state-of-the-art algorithms, showing that it achieves its goal of approximating balanced graph cuts (here: normalized cut) better than the its competitors.

\section{Conclusion and discussion}
\label{sec:outlook}

\noindent In this article we have proposed a novel algorithm \xistref for balanced graph cuts, which preservers the qualitative (multiscale) feature of the original graph cut and meanwhile allows fast computation even for large scale datasets, in particular for sparse graphs. This is achieved by combining the combinatorial nature of $st$-MinCuts with the acceleration techniques based on restriction to vertices of locally maximal degrees and vertex merging. 

We have demonstrated the applicability and versatility of our algorithm on (cell) image segmentation, simulated data as well as large network datasets. The desirable performance of the proposed \xistref, also seen from our simulations, benefits greatly from the structure of $\Vloc$, induced by locally maximizing the degree, and its ordering, which respects the intrinsic geometric structure of the data, to a large extent. We stress, however, that the theoretical findings of \cref{lem:xvst_loc_max_merge_equivalence}, in particular its guarantee to consider $\abs{\Vloc}-1$ distinct partitions, and \cref{lem:xist_complexity} are applicable to any subset of $V$. This means that \xistref can be easily adapted to other interesting subsets of vertices, making the algorithm even more flexible.

\section*{Acknowledgments}
\noindent LS is supported by the DFG (German Research Foundation) under project GRK 2088: \enquote{Discovering Structure in Complex Data}, subproject A1. HL is funded and AM is supported by the DFG under Germany’s Excellence Strategy, project EXC 2067: \enquote{Multiscale Bioimaging: from Molecular Machines to Networks of Excitable Cells} (MBExC). AM and HL are supported by DFG CRC 1456 \enquote{Mathematics of Experiment}, and AM is supported by DFG RU 5381 \enquote{Mathematical Statistics in the Information Age -- Statistical Efficiency and Computational Tractability}, subproject \enquote{Sublinear time methods with statistical guarantees}. The authors would like to especially thank Max Wardetzky (University of Göttingen) for helpful discussions, Ulrike Rölleke and Sarah Köster (University of Göttingen) for providing the cell image dataset, and Florin Manea (University of Göttingen) for pointing out some references.

\printbibliography

\appendix
\section{Appendix}
\label{appendix}

\subsection{Further comparison study}
\label{apdx:sub:kncut_vs_spec_clust}

\cref{apdx:img:kncut_vs_spec_clust} is an extension of \cref{img:kncut_cell_cluster_example}, comparing the \multixistref algorithm and spectral clustering on the same cell image example, for different $k\in\{2,\ldots,9\}$. In particular, this exemplifies the selection process of \multixistref, namely which subgraph to cut.

\begin{figure*}[htpb]
    \begin{multicols}{2}
        \centering
        \subfloat[\multixistref, $k=2$]{\includegraphics[width=0.24\textwidth]{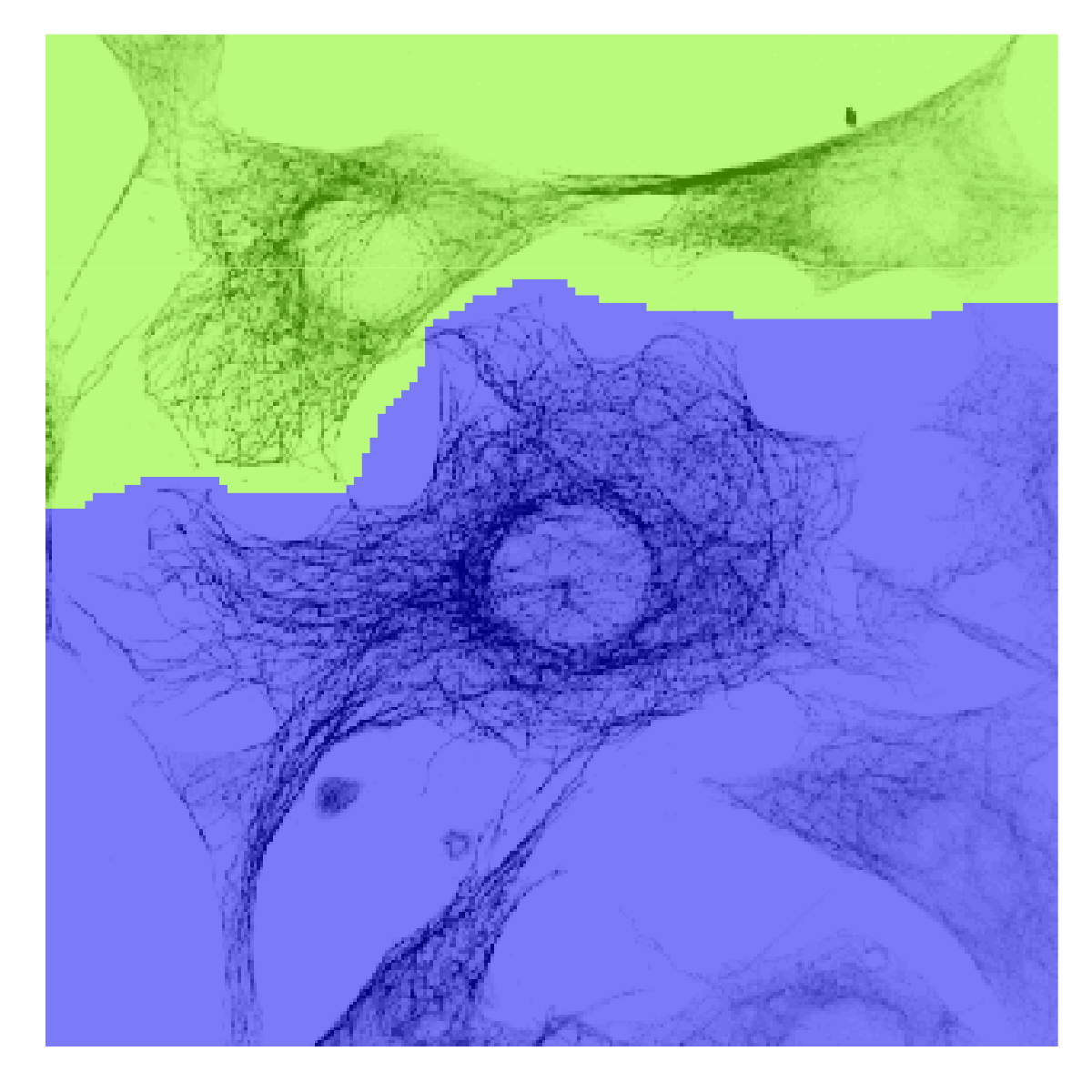}}%
        \subfloat[Spec.\ clust., $k=2$]{\includegraphics[width=.24\textwidth]{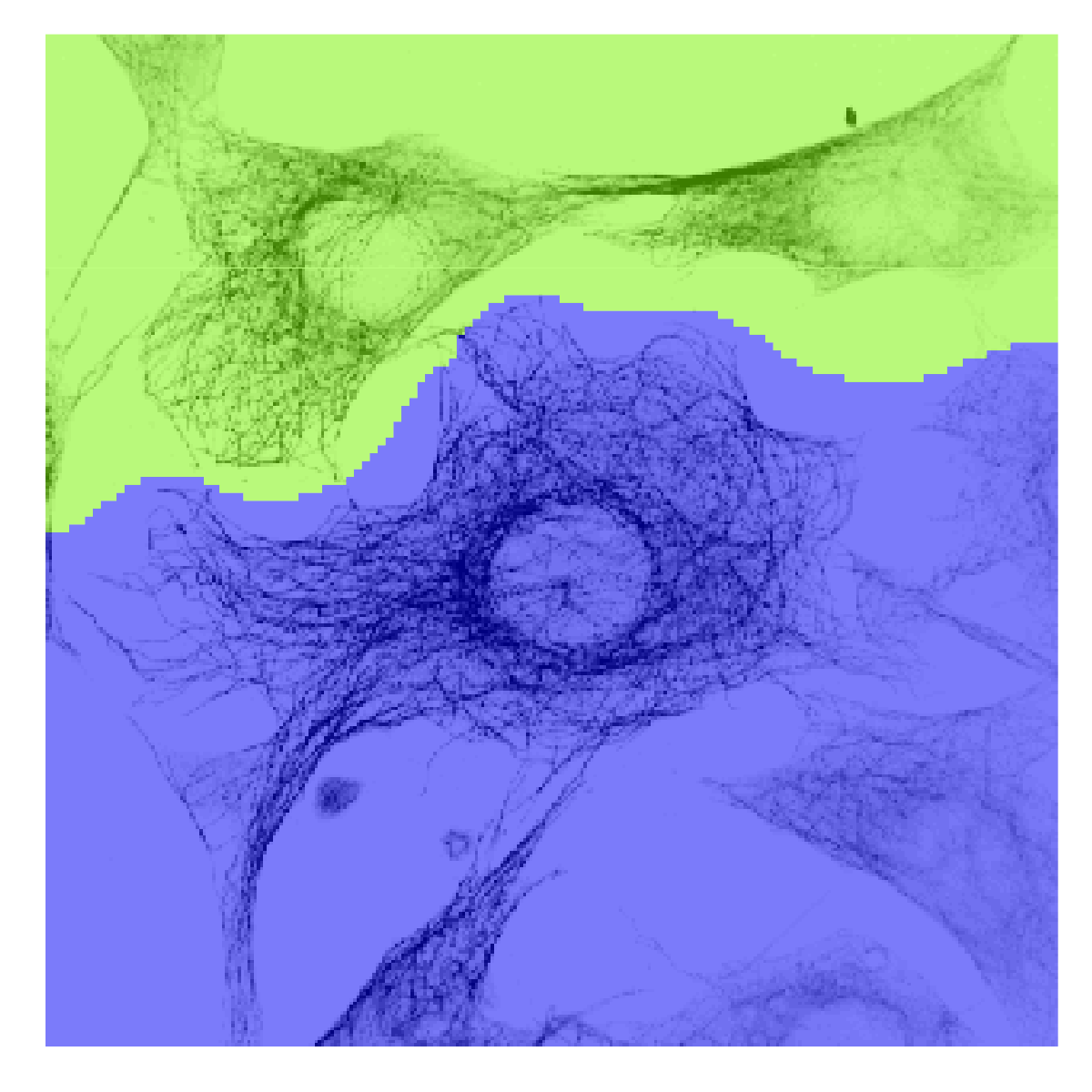}}
        \par
        \subfloat[\multixistref, $k=3$]{\includegraphics[width=0.24\textwidth]{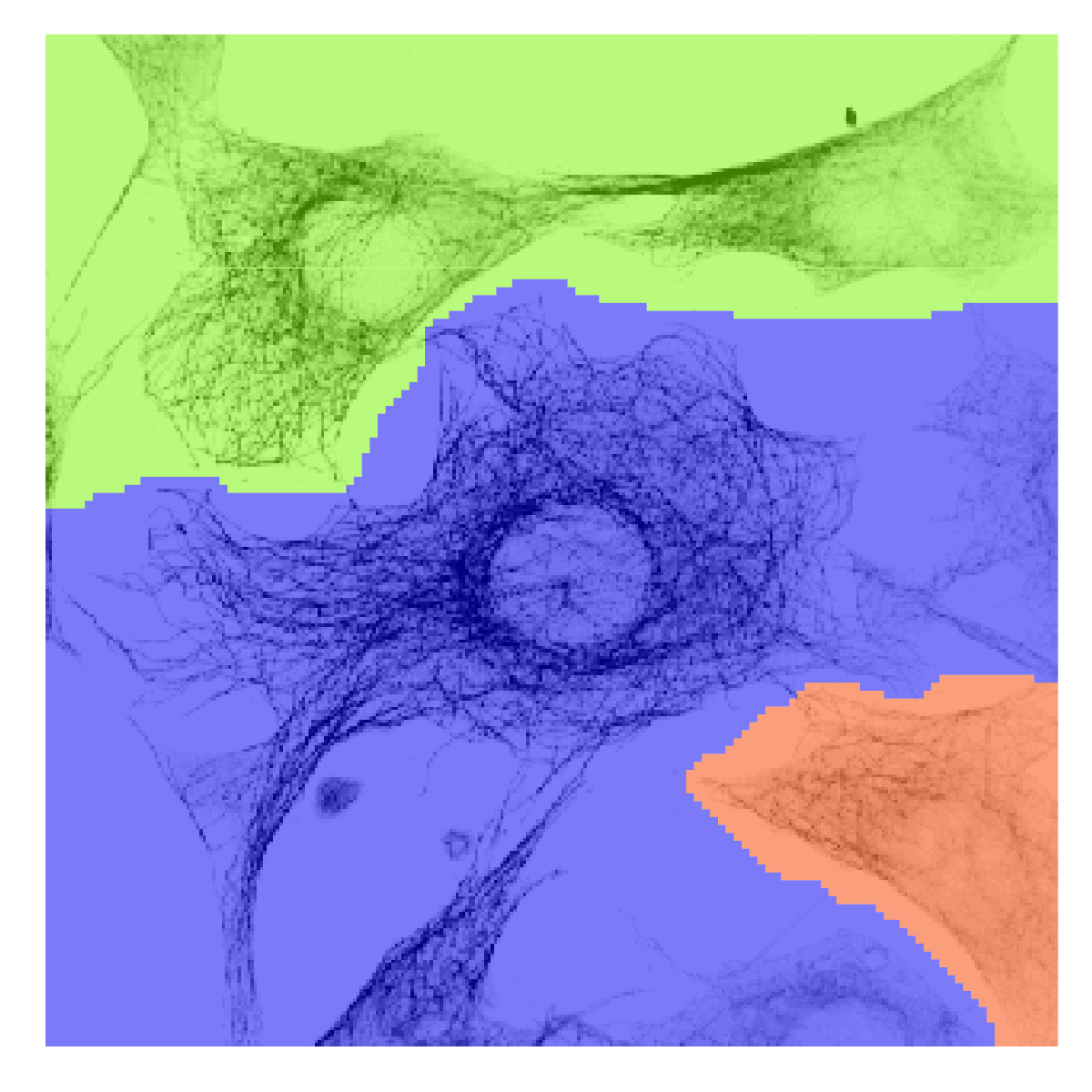}}%
        \subfloat[Spec.\ clust., $k=3$]{\includegraphics[width=.24\textwidth]{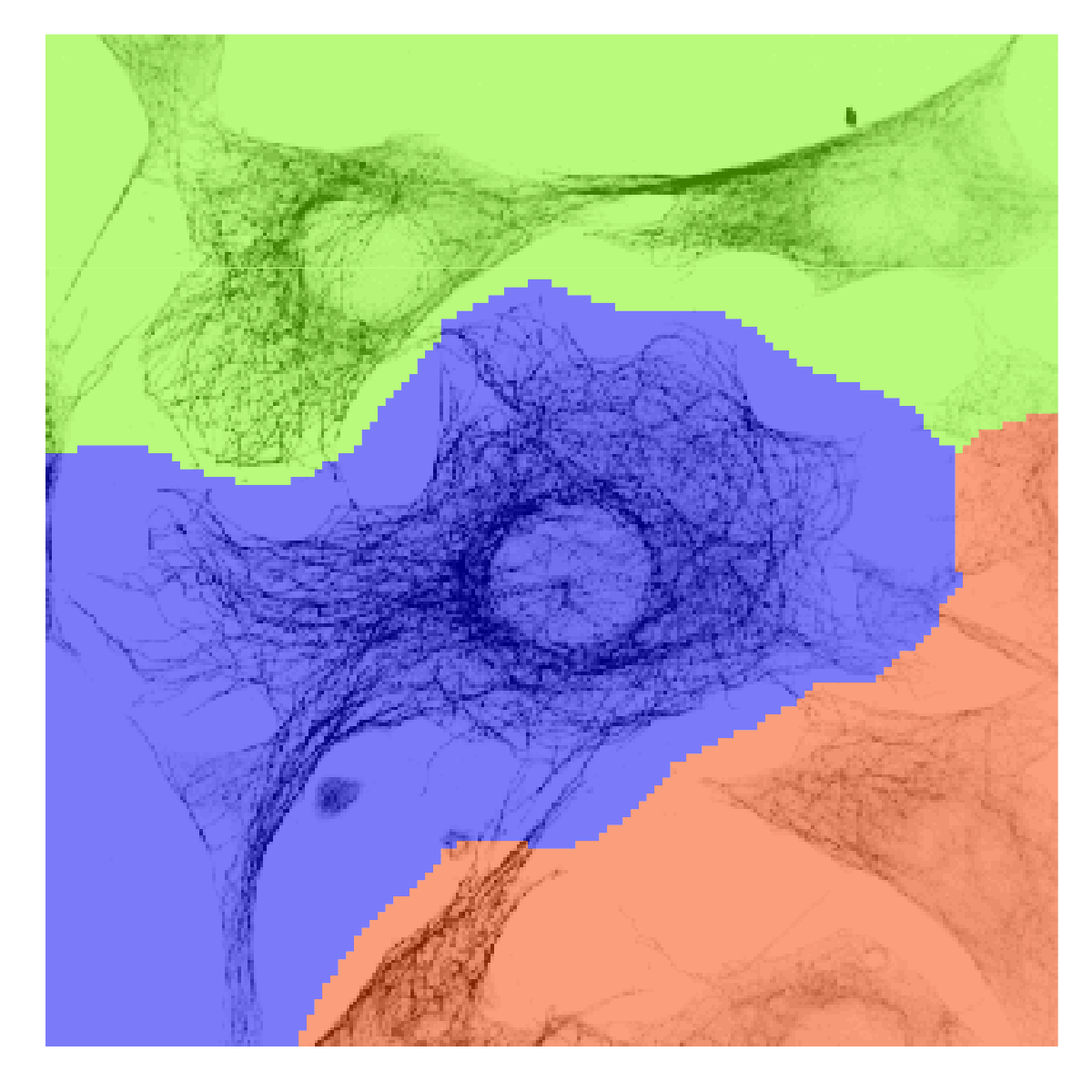}}
        \par
        \subfloat[\multixistref, $k=4$]{\includegraphics[width=0.24\textwidth]{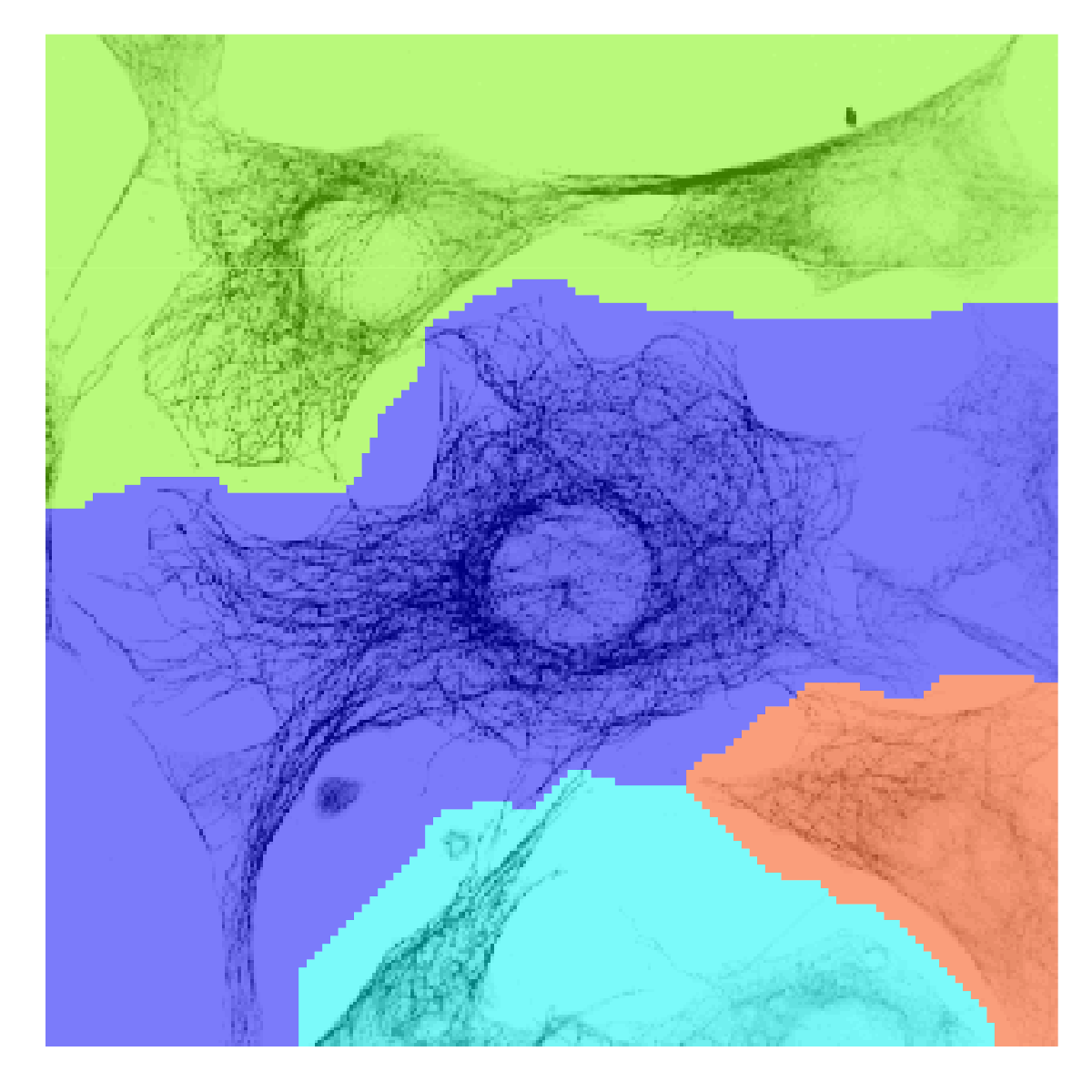}}%
        \subfloat[Spec.\ clust., $k=4$]{\includegraphics[width=.24\textwidth]{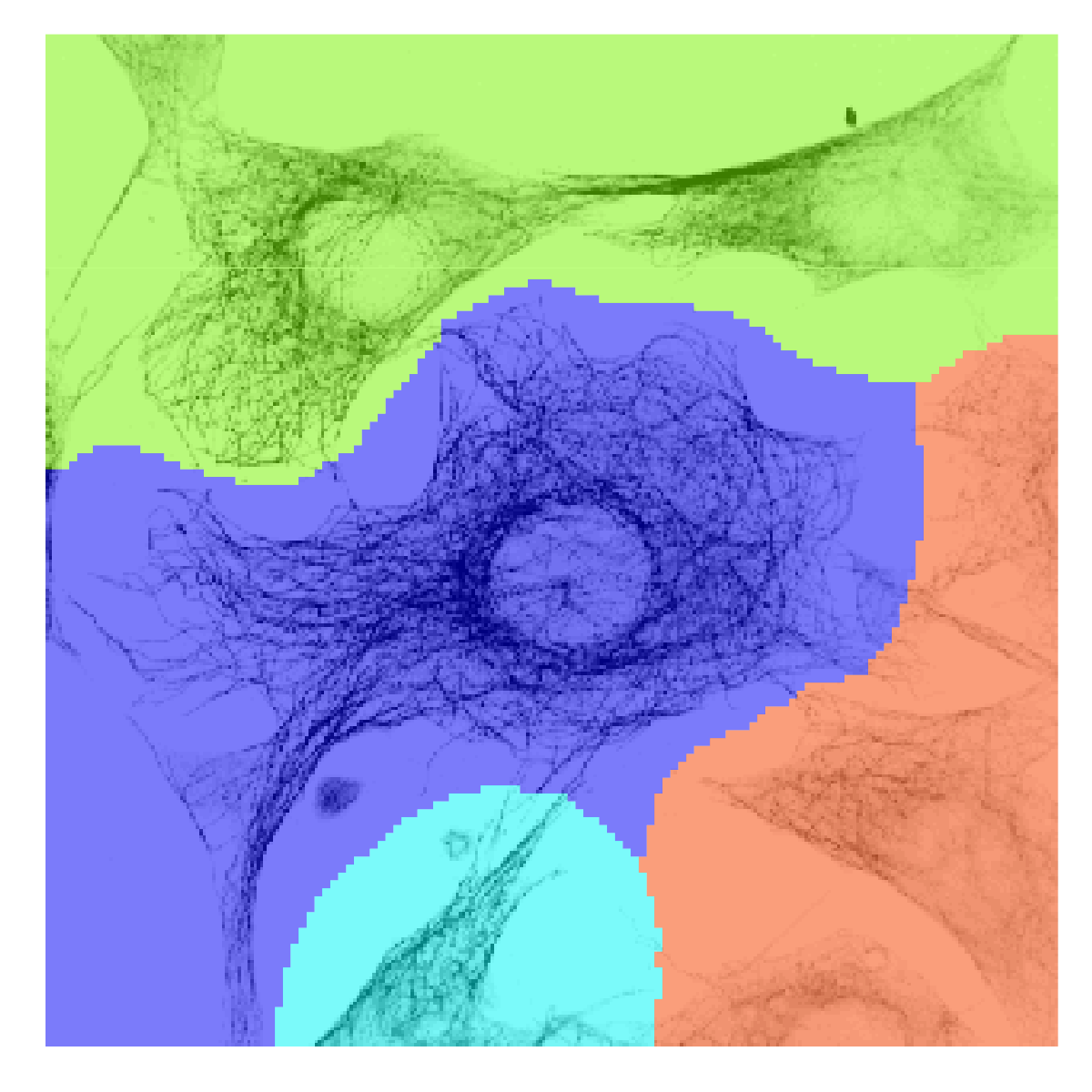}}
        \par
        \subfloat[\multixistref, $k=5$]{\includegraphics[width=0.24\textwidth]{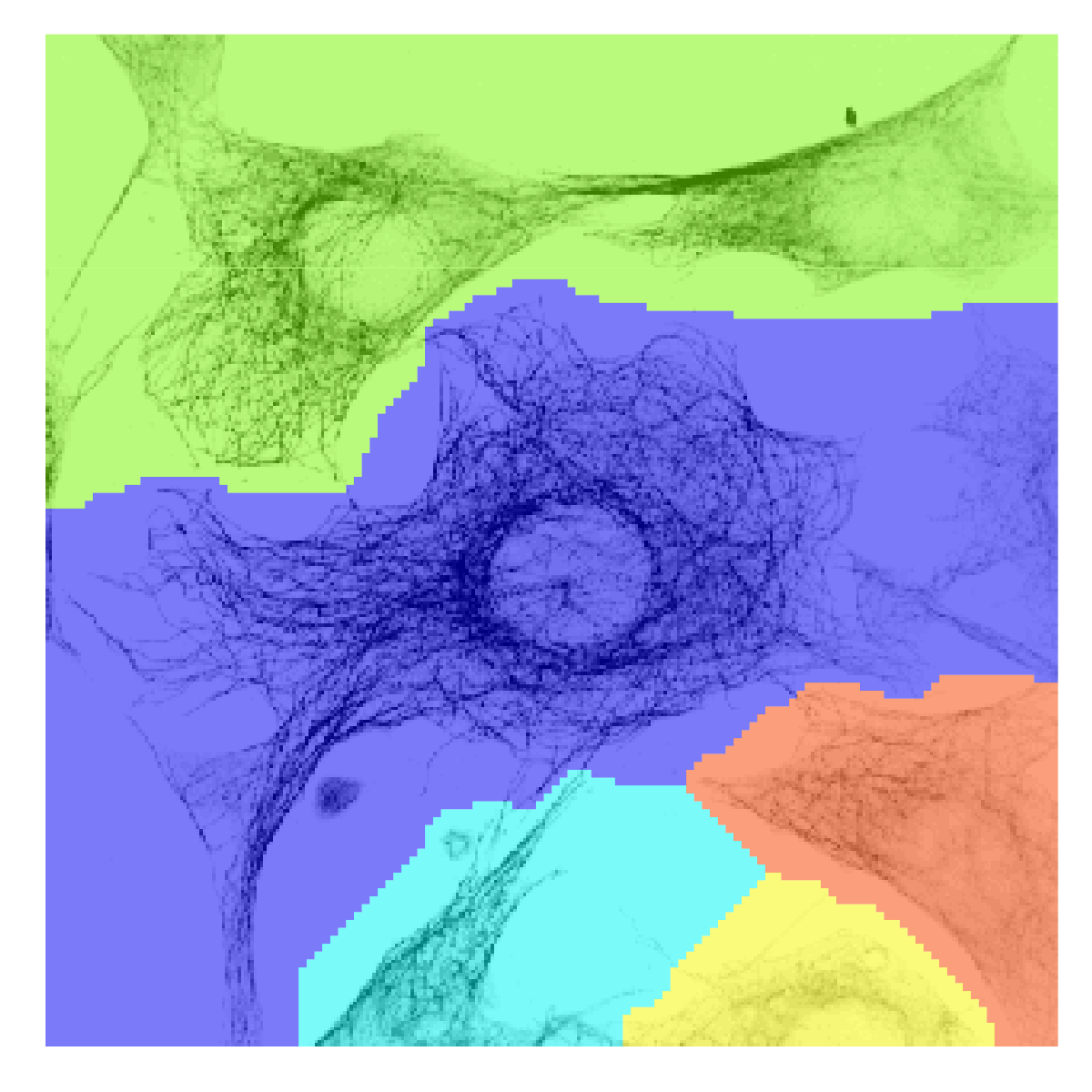}}%
        \subfloat[Spec.\ clust., $k=5$]{\includegraphics[width=.24\textwidth]{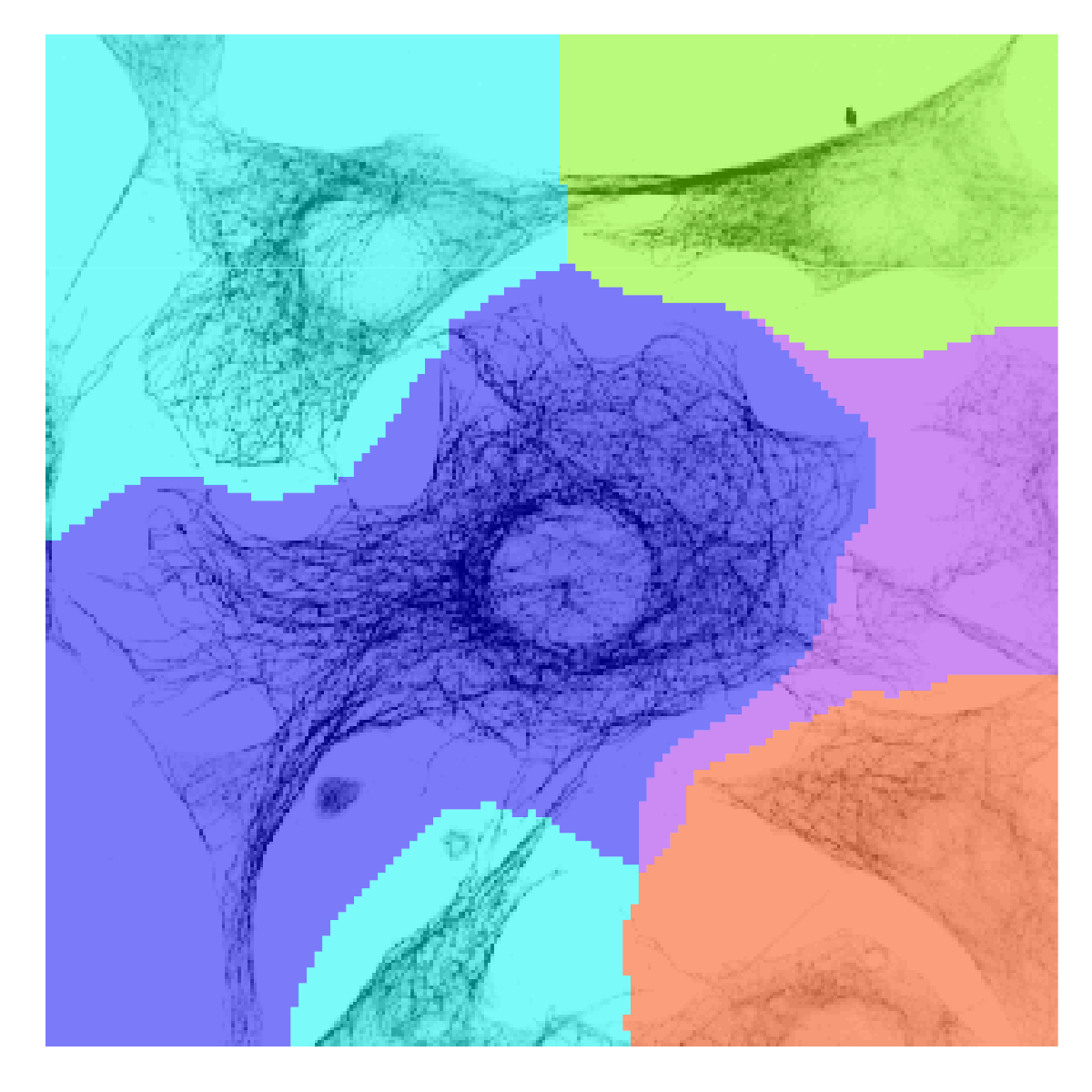}}
        \par
        \subfloat[\multixistref, $k=6$]{\includegraphics[width=0.24\textwidth]{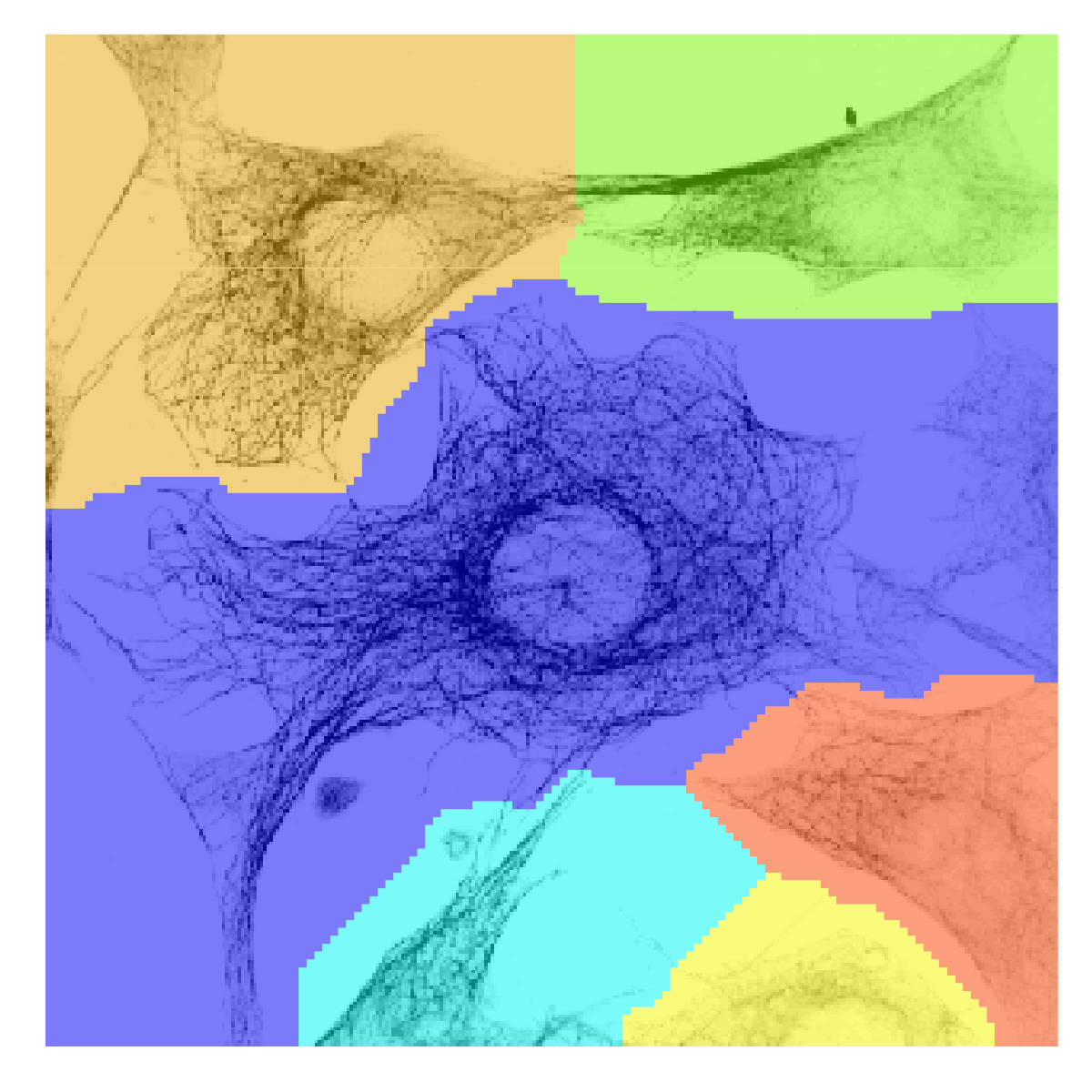}}%
        \subfloat[Spec.\ clust., $k=6$]{\includegraphics[width=.24\textwidth]{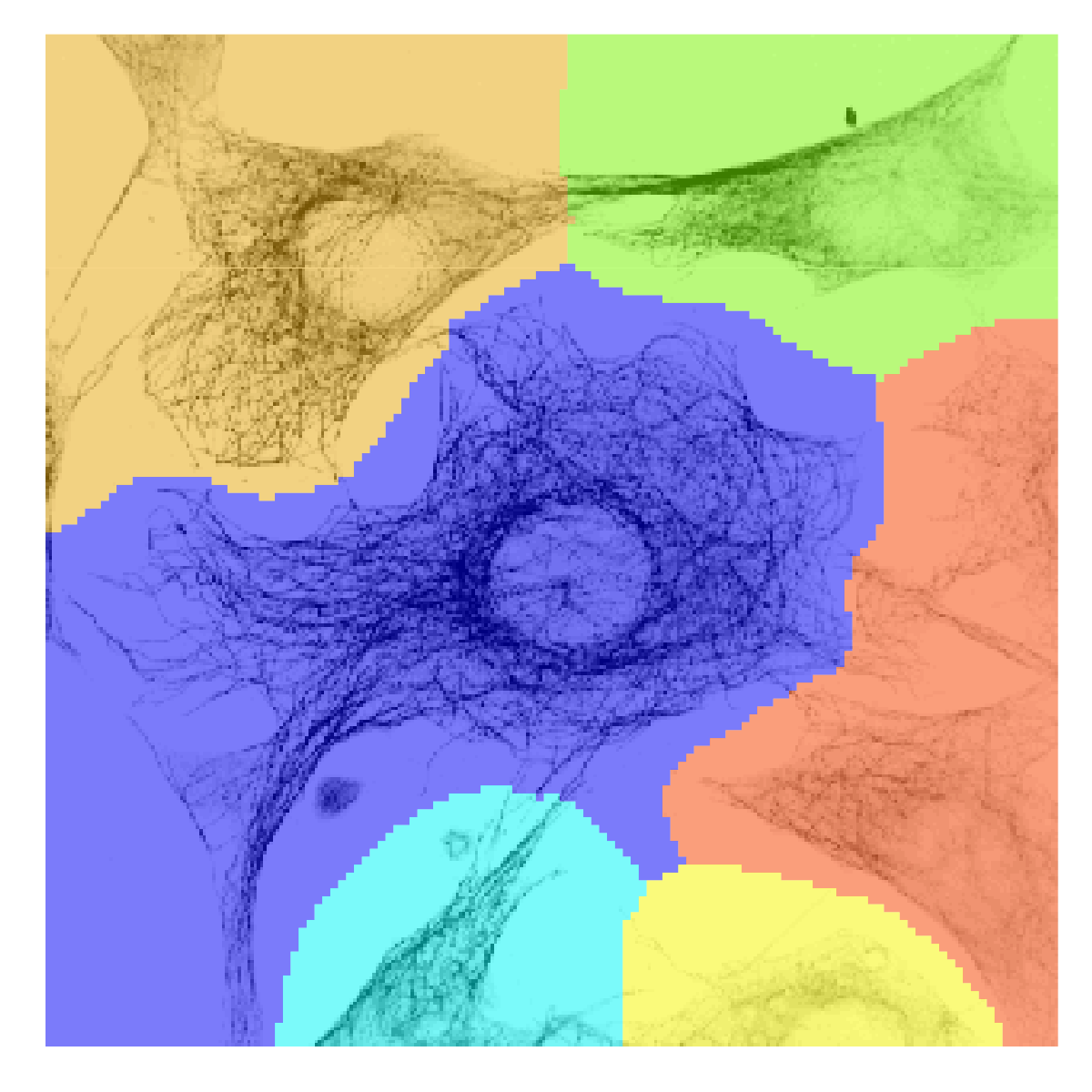}}
        \par
        \subfloat[\multixistref, $k=7$]{\includegraphics[width=0.24\textwidth]{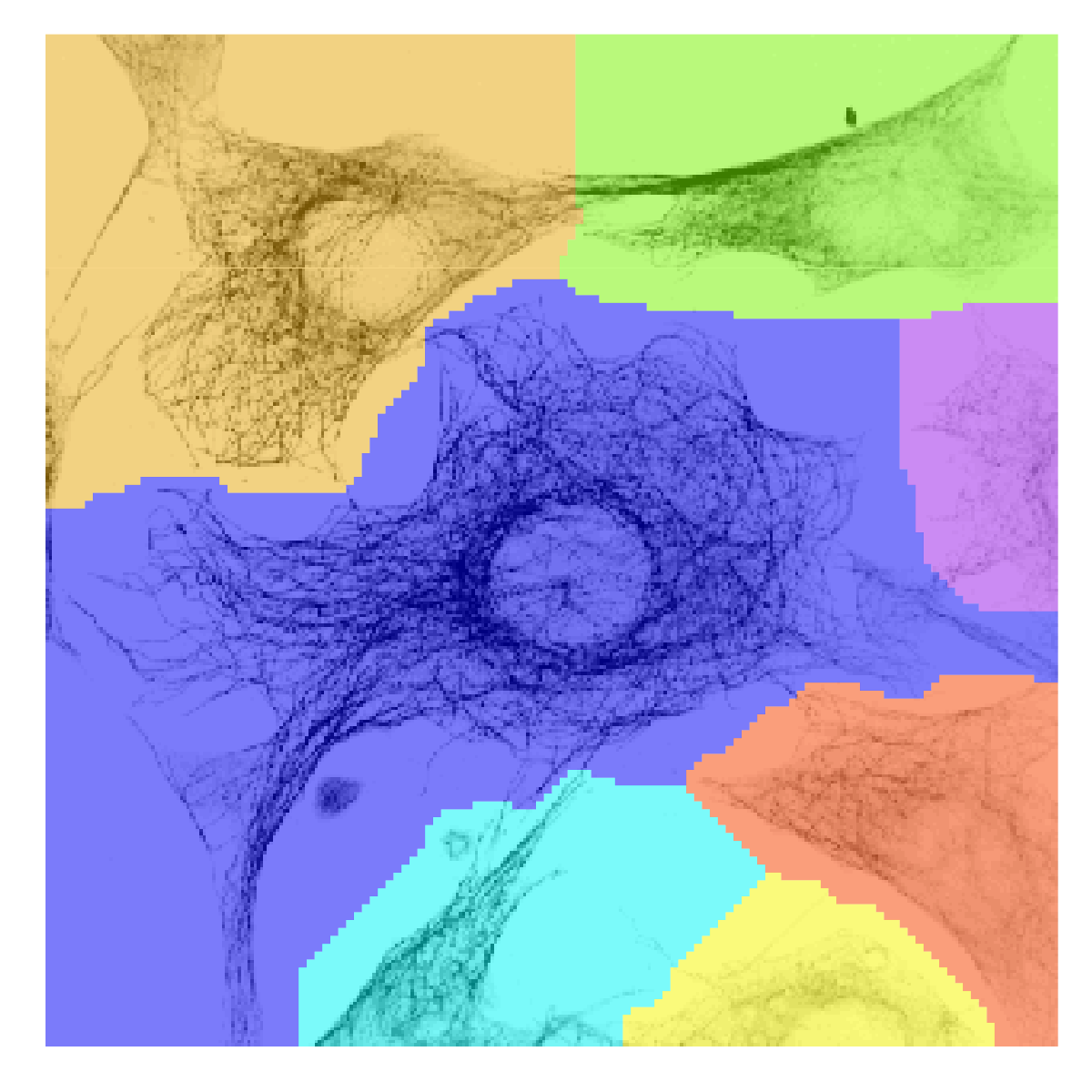}}%
        \subfloat[Spec.\ clust., $k=7$]{\includegraphics[width=.235\textwidth]{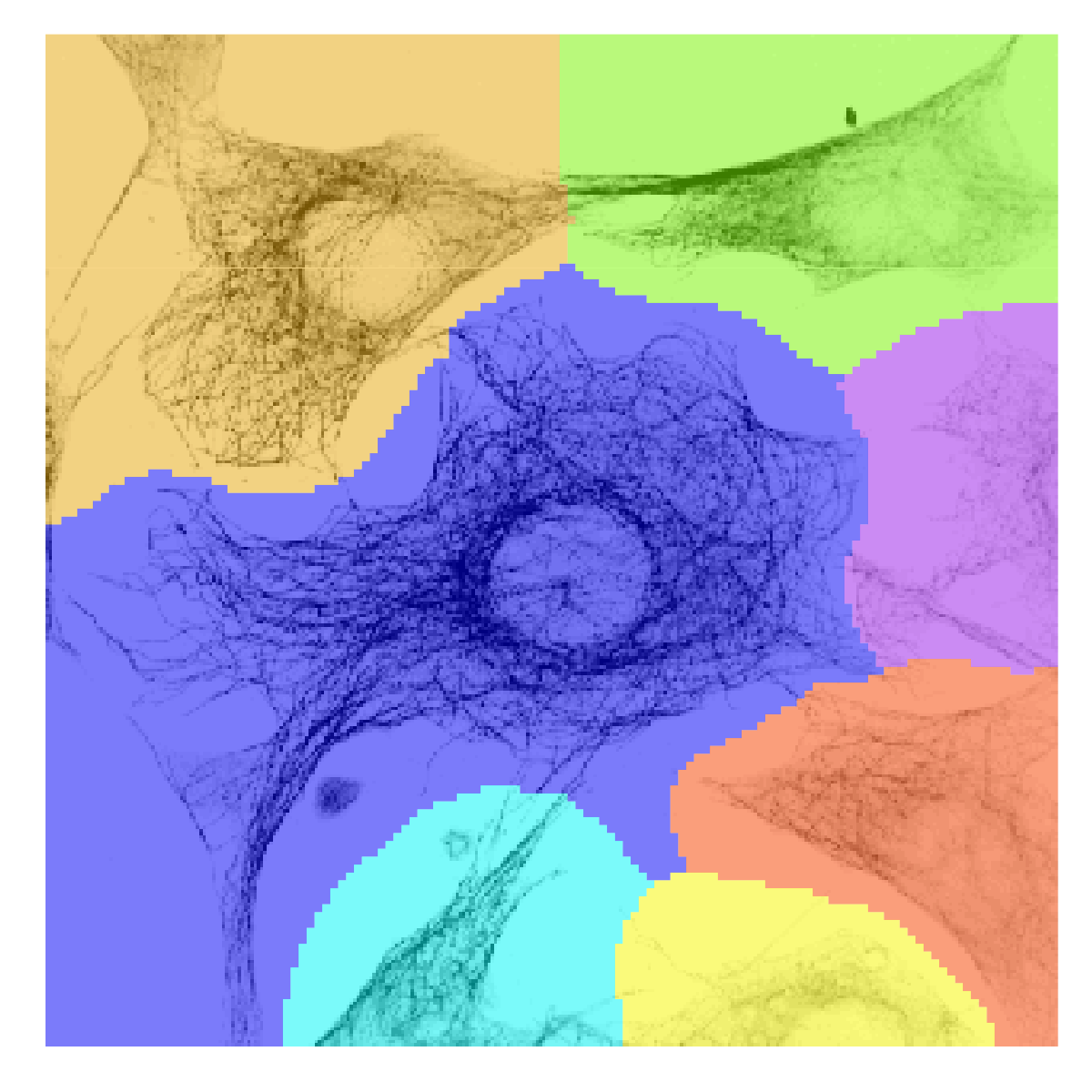}}
        \par
        \subfloat[\multixistref, $k=8$]{\includegraphics[width=0.24\textwidth]{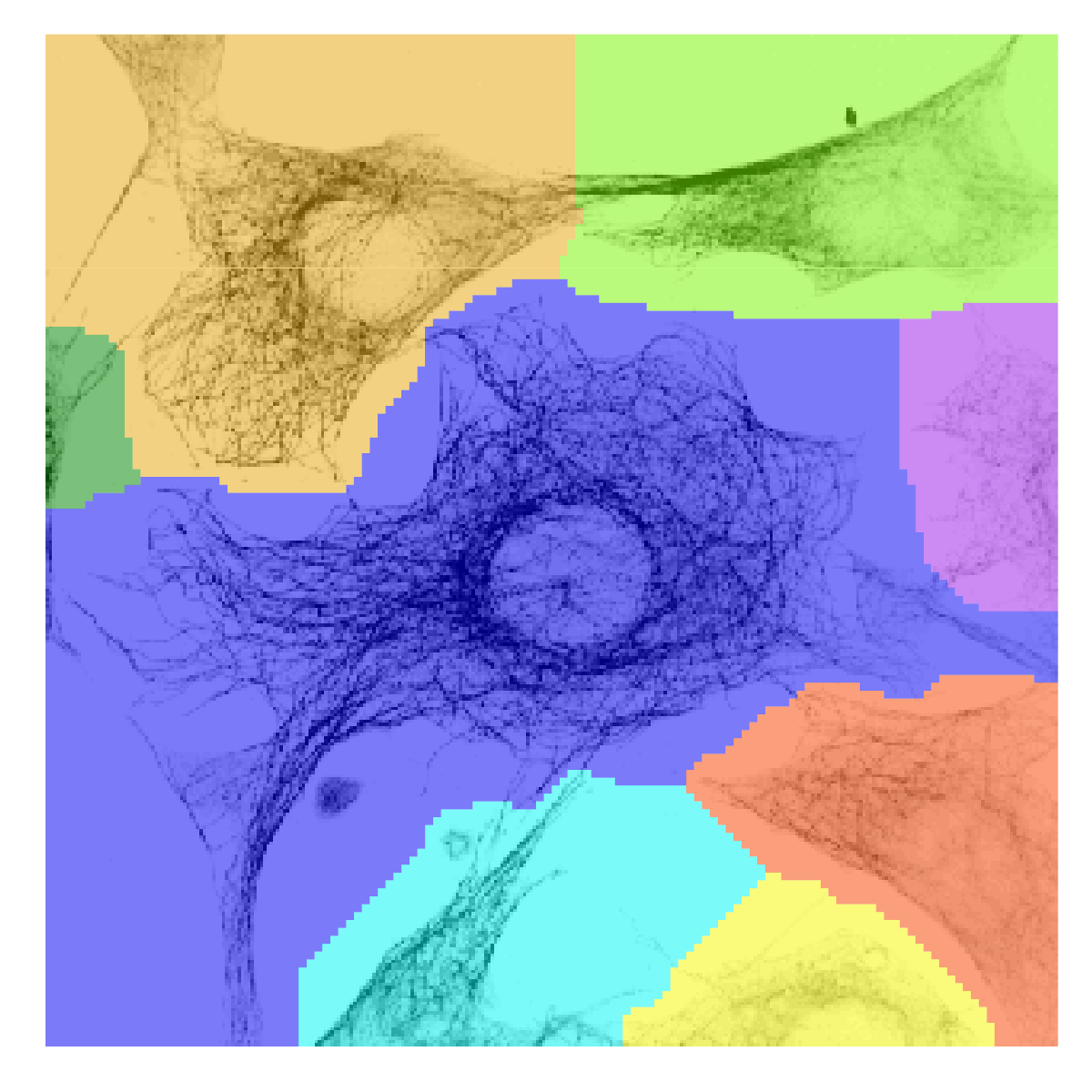}}%
        \subfloat[Spec.\ clust., $k=8$]{\includegraphics[width=.24\textwidth]{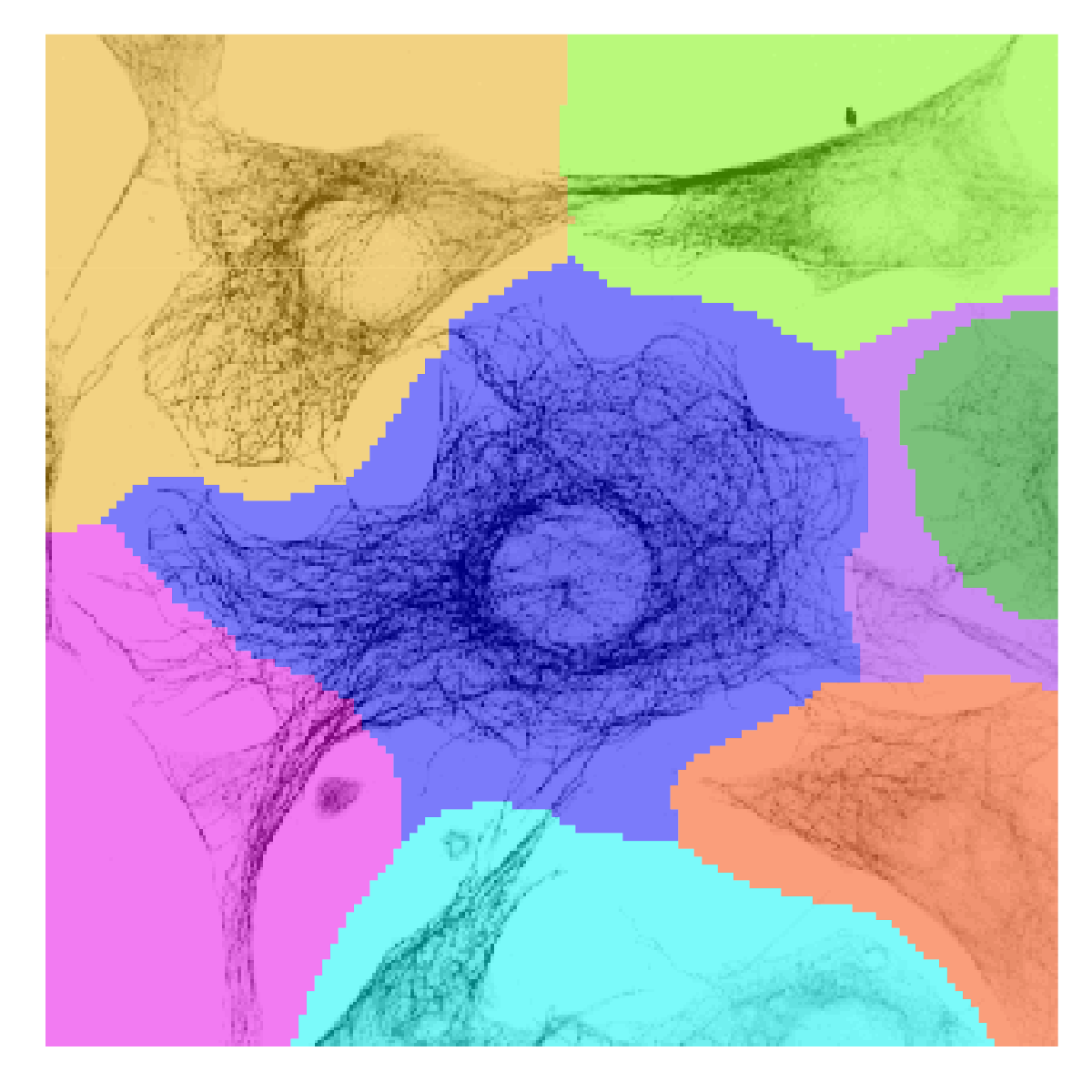}}
        \par
        \subfloat[\multixistref, $k=9$]{\includegraphics[width=0.24\textwidth]{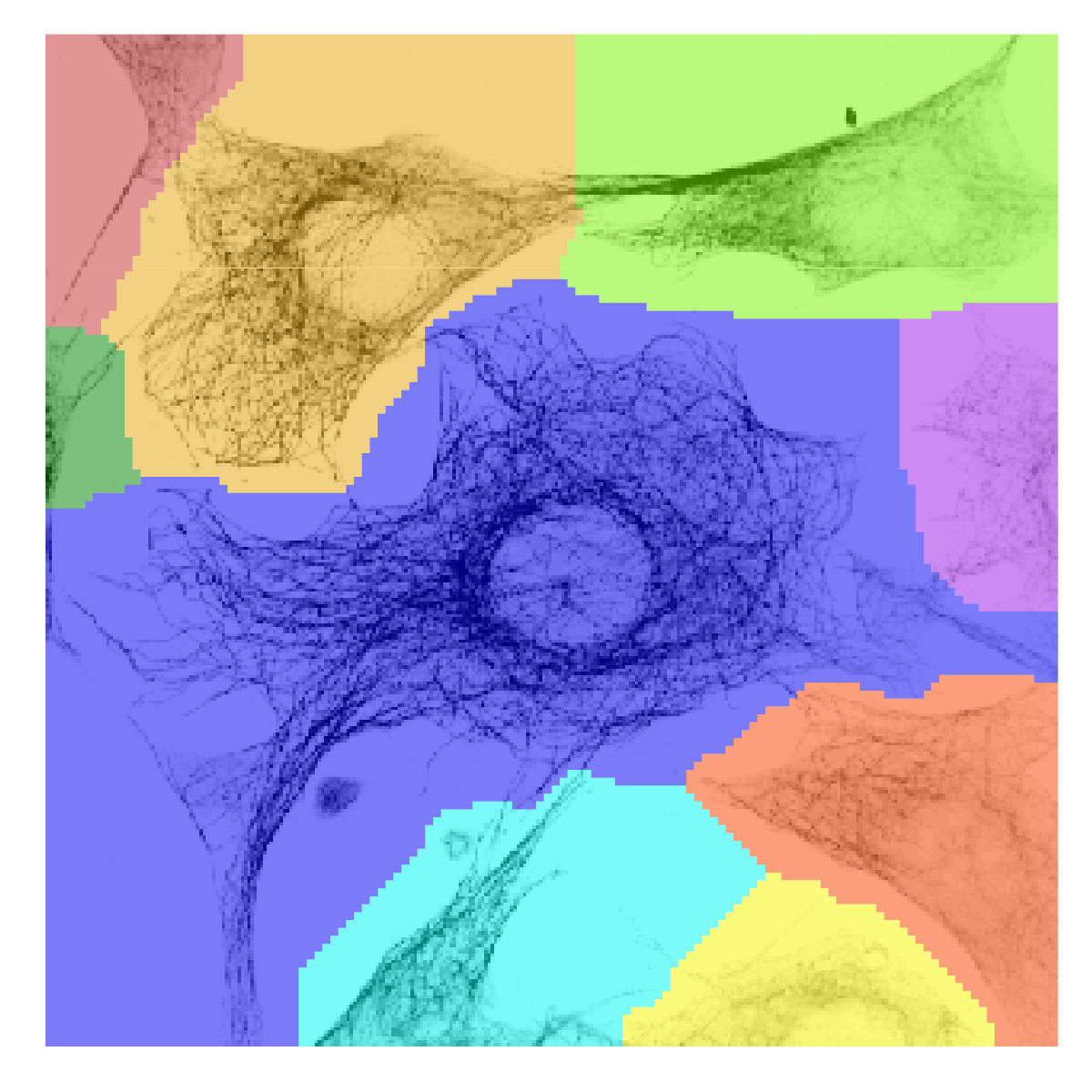}}%
        \subfloat[Spec.\ clust., $k=9$]{\includegraphics[width=.24\textwidth]{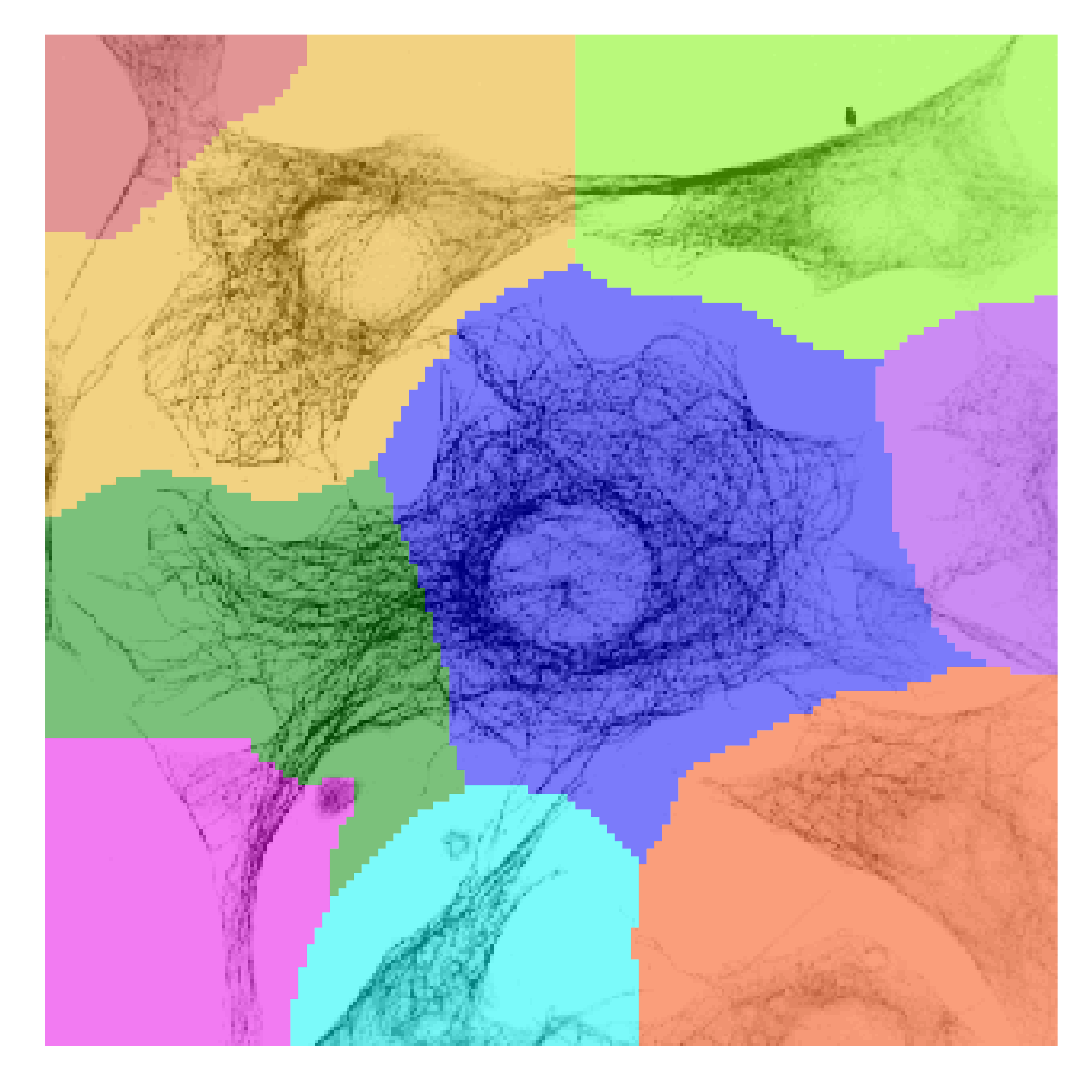}}
    \end{multicols}
    \caption{Extension of \cref{img:kncut_cell_cluster_example} on image segmentation of microtubules in NIH 3T3 cells. The colors visualize the resulting $k\in\{2,\ldots,9\}$ partitions via \multixistref and spectral clustering, and the underlying cell clusters are visible in black. The case $k=10$ is depicted in \cref{img:kncut_cell_cluster_example}.}
    \label{apdx:img:kncut_vs_spec_clust}
\end{figure*}

\subsection{Correctness of \texorpdfstring{\xistref}{Xist}}
\label{apdx:sub:xist_consistency_proof}

In the following we present a proof for the correctness of our algorithms as previously claimed in \cref{lem:xvst_loc_max_merge_equivalence}. As seen in the proof of this theorem, \xistref and the \xvstnameref yield the same output under \cref{assumptions:uniqueness}. Here, we show the remaining claim, namely that \xistref outputs $\min_{s,t\in\Vloc}\XC_{S_{st}}(G)$. Clearly, by definition of \xistref (specifically lines~\ref{alg:xist:ifmin}-\ref{alg:xist:ifmin_end}), we are only left to show that through the selection method via $\tau$, $st$-MinCuts for all pairs $s,t\in A:=\Vloc$ are considered.

We start with some preliminary results. The first of these statements and their proofs are due to \cite{GomoryHu1961}, and the overall structure of this section as well as the remaining claims and the proof of \cref{apdx:thm:xist_correctness} are adapted from \cite{Gusfield1990}, with changes necessary to account for the specific design of \xistref and the generalization to only consider vertices from an arbitrary (but fixed) subset $A\subset V$.

\begin{lemma}
\label{apdx:lem:inequality_stmincuts}
For a weighted graph $G=(V,E,\mat{W})$ and any $v_1,\ldots,v_k\in V$, $k\in\N_{\geq 2}$, with $v_i\neq v_{i+1}$ for $i\in\{1,\ldots,k-1\}$, denote $c_{ij} := \MC_{S_{v_i v_j}}(G)$. Then
\begin{equation}
    c_{1,k}\geq \min\{c_{1,2}, c_{2,3}, \ldots, c_{k-1,k}\}.
\label{apdx:eq:inequality_stmincuts}
\end{equation}
\vspace{-5mm}
\end{lemma}
\begin{proof}
Suppose the claim does not hold, i.e.\ $c_{1,k} < c_{i,i+1}$ for all $i=1,\ldots,k-1$. Let $S\subset V$ denote the partition attaining the $v_1 v_k$-MinCut $c_{1,k}$. Then, there exists an $i\in\{1,\ldots,k-1\}$ such that $v_i$ and $v_{i+1}$ are on different sides of $S$ (i.e.\ either $v_i\in S$ and $v_{i-1}\in\comp{S}$ or vice versa) because $v_1\in S$ and $v_k\in\comp{S}$. Then, for this $i$, the partition attaining $c_{i,i+1}$ is also a valid $v_1 v_k$-cut, and thus, as $c_{1k}$ is the $v_1 v_k$-MinCut value, $c_{1k}\geq c_{i,i+1}$, contradicting the supposition.
\end{proof}

\begin{corollary}
For a weighted graph $G=(V,E,\mat{W})$ and pairwise distinct $s,t,v\in V$, $\min\{\MC_{S_{st}}(G), \MC_{S_{sv}}(G), \MC_{S_{vt}}(G)\}$ is not uniquely attained.
\label{apdx:cor:three_mincut_nonuniqueness}
\end{corollary}

The following \cref{apdx:lem:noncrossing_cut_existence} shows that for any two vertices $u,v$ on the same side of an $st$-MinCut there is a $uv$-MinCut that preserves the $st$-MinCut partition by remaining on one side of the cut.

\begin{lemma}
Let $G=(V,E,\mat{W})$ be a weighted graph, let $s,t\in V$, $s\neq t$, and $u,v\in S_{st}$, $u\neq v$. Then, if $t\in \comp{S}_{uv}$, $S_{st}\cap S_{uv}$ is a $uv$-MinCut, and if $t\in S_{uv}$, $S_{st}\cap\comp{S}_{uv}$ is a $uv$-MinCut.
\label{apdx:lem:noncrossing_cut_existence}
\end{lemma}
\begin{proof}
For the $uv$-MinCut partition $S_{uv}$ in $G$, define
\[
A_1 := S_{st}\cap S_{uv},\quad A_2 := S_{st}\cap \comp{S}_{uv},\quad A_3 := \comp{S}_{st}\cap S_{uv},\text{ and } A_4 := \comp{S}_{st}\cap \comp{S}_{uv}.
\]
W.l.o.g.\ assume that $s,u\in A_1$ (otherwise exchange $u$ and $v$ and/or $A_1$ and $A_3$), and that $v\in A_2$ (otherwise exchange $A_2$ and $A_4$). For $k,\ell\in\{1,2\}$ define
\[
c_{k\ell} := \sum_{\substack{i\in A_k, j\in B_{\ell} \\ \{i,j\}\in E}} w_{ij}.
\]
First assume that $t\in A_3$. Notice that $A_2$ is a valid $uv$-cut in $G$, and that $S_{uv} = A_1\cup A_3$ is the $uv$-MinCut in $G$, so
\begin{align}
0 &\leq \MC_{A_2}(G) - \MC_{S_{uv}}(G) \notag\\
&= (c_{21} + c_{23} + c_{24}) - (c_{12} + c_{14} + c_{32} + c_{34}) \notag\\
&= c_{24} - (c_{14} + c_{34})
\label{apdx:eq:cases_cuts_1}
\end{align}
since $c_{k\ell} = c_{\ell k}$ for all $k,\ell=1,2$. Now assume first that $t\in A_3$. Then, $A_3$ is a valid $st$-cut in $G$, and $S_{st} = A_1\cup A_2$ is the $st$-MinCut in $G$, so that
\begin{align}
0 &\leq \MC_{A_3}(G) - \MC_{S_{st}}(G) \notag\\
&= (c_{31} + c_{32} + c_{34}) - (c_{13} + c_{14} + c_{23} + c_{24}) \notag\\
&= c_{34} - (c_{14} + c_{24})
\label{apdx:eq:cases_cuts_2}
\end{align}
Adding \eqref{apdx:eq:cases_cuts_1} and \eqref{apdx:eq:cases_cuts_2} yields $0\leq -2 c_{14}$ and thus $c_{14} = 0$. Using \eqref{apdx:eq:cases_cuts_1} (or \eqref{apdx:eq:cases_cuts_2}) again gives $c_{24} = c_{34}$, so that
\[
\MC_{A_2}(G) = c_{21} + c_{23} + c_{24} = c_{12} + c_{14} + c_{32} + c_{34} = \MC_{S_{uv}}(G),
\]
i.e.\ $A_2$ also is a $uv$-MinCut in $G$. As $A_2$ is not affected by the contraction of $\comp{S}_{st}$, the claim follows under the assumption that $t\in A_3$.

If, on the other hand, $t\in A_4$, we proceed analogously to before: Modify \eqref{apdx:eq:cases_cuts_1} by noting that $A_1$ is a valid $uv$-cut in $G$, resulting in $0\leq c_{13} - (c_{23} + c_{34})$, and further modify \eqref{apdx:eq:cases_cuts_2} by noting that now $A_4$ is a valid $st$-cut in $G$, yielding $0\leq c_{34} - (c_{13} + c_{23})$, so that together $c_{34} = 0$ and consequently $c_{13} = c_{23}$. Finally, this results in
\[
\MC_{A_1}(G) = c_{12} + c_{14} + c_{13} = c_{12} + c_{14} + c_{32} + c_{34} = \MC_{S_{uv}}(G),
\]
yielding the claim also in the case of $t\in A_4$.
\end{proof}

In the following we present a short and direct proof that the \xistref algorithm indeed considers $st$-MinCut partitions for all $s,t\in\Vloc$ and, more significantly, that under the mild \cref{assumptions:uniqueness}, the \xvstnameref and \xistref yield the same output. To that end we closely follow the proofs of \citet[Section 2.1]{Gusfield1990}, adapting them to our case whenever necessary, in particular to accommodate for the fact that we are interested in the $st$-MinCut partitions (and not only the cut values themselves) and also the fact that we restrict our attention to only a subset $A\subseteq V$ of vertices (which can be arbitrary for the proof, but for our purposes will of course be $\Vloc$).

To build intuition, we can consider the vector $\tau$ in the \xistref algorithm as representation of a tree -- in fact, in \cref{apdx:thm:xist_correctness} it is shown that this is the Gomory-Hu tree. The interpretation is very simple: There is an edge between $i$ and $\tau_i$ for every $i=1,\ldots,\abs{A}$, where $A\subseteq V$ is the subset of vertices (on $G=(V,E,\mat{W})$) \xistref is applied to, so $\tau$ can be thought of as a tree with branches between $i$ and $\tau_i$. In the following we abbreviate $N:=\abs{A}$ and let $A=\{1,\ldots,N\}\subseteq V:=\{1,\ldots,n\}$ to ease notation. Call $u,v\in A$ \emph{neighbours} at some point in \xistref the $uv$-MinCut is computed (i.e.\ if $v=\tau_u$). Further, for any $u,v\in A$, a sequence of neighbours $\overline{uv} := \{u,v_1,\ldots,v_k,v\}$ such that $u=\tau_{v_1}$, $v_k = \tau_v$ and $v_i = \tau_{v_{i+1}}\in A$ for all $i=1,\ldots,k-1$, for some $k\leq N$, is called a \emph{directed path} from $u$ to $v$.

\begin{lemma}
For a weighted graph $G=(V,E,\mat{W})$ apply \xistref to an $A\subseteq V$ with $s,t\in A$, $s\neq t$, such that $s$ and $t$ are connected by a directed path $\overline{st}$, and let $v\in A$ be connected by a directed path to $t$ such that $v < u$ for any $u\in\overline{st}\setminus\{t\}$. Then $s\in S_{vt}$ if and only if $v\in\overline{st}$.
\label{apdx:lem:gusfield_rules}
\end{lemma}
\begin{proof}
At the beginning of \xistref, $\tau_s=1$, and, from iterations $2$ to $s-1$ in line~\ref{alg:xist:forloop}, $\tau_s$ changes from one $j$ to one $\ell$, $j,\ell\in\{1,\ldots,N\}$, if and only if (between iterations $2$ and $i-1$) $j$ and $\ell$ were \emph{neighbours} (i.e.\ the cut between $j$ and $\tau_{j}=\ell$ is computed at some point throughout \xistref). Consequently, a node $j\in A$ was a neighbour of $s$ at some point if and only if $j\in\overline{1s}$. As $t < v$, $v$ must have been a neighbour of $s$ before computation of the $vt$-cut. Moreover, as $v < u$ for any $u\in\overline{st}\setminus\{t\}$, $t$ is a neighbour of $s$ throughout the computation of the $vt$-MinCut (with partition $S_{vt}$). Thus, if $v\in\overline{st}$, $s\in S_{vt}$, and $s\notin S_{vt}$ otherwise.
\end{proof}

\begin{theorem}
For a weighted graph $G=(V,E,\mat{W})$ and any subset $A\subseteq V$, \xistref applied to $G$ and $A$ considers $st$-MinCuts for all pairs $s,t\in A$.
\label{apdx:thm:xist_correctness}
\end{theorem}
\begin{proof}
First, note that for each $i=2,\ldots,N$, \xistref computes an $i\tau_i$-MinCut since always $\tau_i < i$ holds, and thus, after an $i\tau_i$-MinCut is computed in step $i$, $\tau_i$ is not changed afterwards.

Let now $s,t\in A$, $s\neq t$, be arbitrary, and, depending on the context, consider the (directed) path $\overline{st}$ either as a sequence of vertices $s,v_1,\ldots,v_k,t$ or edges $(s,v_1),(v_1,v_2),\ldots,(v_{k-1},v_k),(v_k,t)$. Abbreviating $c_{uv} := \MC_{S_{uv}}(G)$ and $C_{uv} := \{c_{ij}\mid (i,j)\in\overline{uv}\}$ (for any $u,v\in V$, $u\neq v$), we will show that 
$$
c_{st} = \min C_{st}
$$
which yields the claim immediately because \xistref computes all cuts in $\overline{st}$. By \cref{apdx:lem:inequality_stmincuts}, \enquote{$\geq$} holds. For the opposite direction, suppose that \enquote{$\leq$} does not hold, and let $s,t\in A$, $s\neq t$, be the vertices forming the shortest path in $\tau$ with this property, i.e.\ such that $c_{st} > \min C_{st}$.

First, assume that $s$ and $t$ are connected by a directed path in $\tau$, i.e.\ $s=\tau_{v_1}$, $v_k = \tau_t$ and $v_i = \tau_{v_i+1}$ for all $i=1,\ldots,k-1$. As the case $k=0$ was tackled in the beginning of this proof, $v:=v_k\in A\setminus\{s,t\}$. Due to $\overline{st}$ being minimal (and \cref{apdx:lem:inequality_stmincuts}), $c_{sv} = \min C_{sv}$, so $c_{sv} = c_{vt} = \min C_{st}$ by \cref{apdx:cor:three_mincut_nonuniqueness} and the supposition that $c_{st} > \min C_{st}$. However, by \cref{apdx:lem:gusfield_rules}, $S_{vt}$ is also a valid $st$-cut, so $c_{st}\leq c_{vt} = \min C_{st}$, leading to a contradiction.

Second, we tackle the case where $s$ and $t$ are not connected by a directed path, but rather that there is an $r\in A$ such that $\overline{sr}$ and $\overline{rt}$ each form a directed path in $\tau$. As by the first case, $c_{sr} = \min C_{sr}$ and $c_{rt} = \min C_{rt}$ and thus, again, by \cref{apdx:cor:three_mincut_nonuniqueness} and the supposition we have $c_{sr} = c_{rt} = \min C_{st} = \min C_{sr}\cup C_{rt}$. Define $(u,v)\in\overline{sr}$ to be the vertices closest to $r$ with $c_{uv} = \min C_{sr} = c_{sr}$. By \cref{apdx:lem:gusfield_rules}, $s\in S_{uv}$ and, since $c_{st} > \min C_{st} = \min C_{sr}$, $t\in S_{uv}$.

Let $r_1,r_2\in A$ such that $(r_1,r)\in\overline{sr}$ and $(r,r_2)\in\overline{rt}$. W.l.o.g.\ assume that \xistref computed the $r_1 r$-MinCut before the $r r_2$-MinCut. Then $s,u\in S_{r_1 r}\cap S_{uv}$, $v\in S_{r_1 r}\cap \comp{S}_{uv}$ and $t\in \comp{S}_{r_1 r}\cap S_{uv}$. We can now apply \cref{apdx:lem:noncrossing_cut_existence} to the cuts $S_{r_1 r}$ and $S_{uv}$. 
This leads to two cases: If $r\in S_{uv}$, $S_{r_1 r}\cap S_{uv}$ forms an $uv$-MinCut that is also a valid $vr$-cut, and if $r\in\comp{S}_{uv}$, $S_{r_1 r}\cap \comp{S}_{uv}$ constitutes an $uv$-MinCut that is also a valid $st$-cut. In the second case, by definition of $u$ and $v$, $c_{st}\leq c_{uv} = \min C_{st}$, showing the claim. In the first case, $c_{vr}\leq\min C_{st}$, but, since $v$ and $r$ are connected by a directed path in $\tau$, by the above considerations $c_{vr} = \min C_{vr}$ as well as $\min C_{vr} > \min C_{st}$ since $(u,v)$ was the closest pair to $r$ that attains the minimum $\min C_{st}$. Together, $c_{vr} > \min C_{st}$, though this immediately contradicts the definition of $(u,v)$, finishing the proof.
\end{proof}

\subsection{Proof of \texorpdfstring{\cref{lem:xist_complexity}}{Lemma \ref{lem:xist_complexity}}}
\label{apdx:xist_complexity_proof}

\begin{proof}
Since an entry $w_{ij}$ of $\mat{W}$ is zero if $\{i,j\}\not\in E$, we only need $V$ stored as a \emph{list} (or a \emph{self-balancing binary search tree}), and $\mat{W}$ stored as an \emph{array}, as input for the \xvstnameref and \xistref. We first analyze \xistref line by line:
\begin{enumerate}
    \item[1] $\BO(1)$ computations.
    \item[2] Determining $\Vloc$ requires computing the degree of every vertex and determining the local maxima. This can be done simultaneously, and takes at most $n^2$ additions and $n-1$ comparisons while taking $\BO(n^2)$ time.
    \item[3] Determining $N$ takes $\BO(N)$ computations, but only determining whether the algorithm terminates here requires only $\BO(1)$ computations.
    \item[4] $\BO(N)$ runtime due to the creation of $\tau$.
    \item[5--14] Because of line~\ref{alg:xist:forloop} the following steps are executed exactly $N-1$ times:
    \begin{enumerate}
        \item[6] $\BO(1)$ computations.
        \item[7] Orlin's algorithm \cite{Orlin2013} for computing an $st$-MinCut partition (by solving the dual problem of computing a max flow from $s$ to $t$) takes $\BO(nm+m^{31/16}\log^2 n)$, which can be reduced for specific types of graphs, e.g.\ $m=\BO(n)$ yields $\BO(n^2/\log(n))$. If $m=\Omega(n^{1+\eps})$, \cite{KingRaoTarjan1994} provided an $\BO(nm)$ algorithm which was recently improved by \cite{OrlinGong2021} who obtained $\BO\Bigl(\tfrac{nm\log n}{\log\log n + \log\tfrac{m}{n}}\Bigr)$ for $m\geq n$. 
        Thus, by \cite{KingRaoTarjan1994,Orlin2013,OrlinGong2021}, an $st$-MinCut partition can be computed in $\BO(nm)$ time for all possible values of $m$ and $n$.
        \item[9] By assumption, the computation time of the XC value for partition $\comp{S}_{st}$ is $\BO(\kappa)$.
        \item[10--13] $\BO(1)$ computations.
        \item[14-17] Iterating through all $j=1,\ldots,N$ takes $\BO(N)$, and for each $j$, only $\BO(1)$ computations are done, yielding $\BO(N)$ in total.
    \end{enumerate}
\end{enumerate}

Therefore, \xistref runs in $\BO(N\max\{nm,\kappa\})$ time for general $m$, and in $\BO(N\max\{n^2/\log n, \kappa\})$ time for $m=\BO(n)$ by utilizing \cite{Orlin2013} as noted above. As the \xvstnameref iterates over all pairs of vertices in $V$ in lines~\ref{alg:xvst:while_loop} to \ref{alg:xvst:while_end}, to ascertain its complexity we simply substitute $N$ by $n^2$ to obtain $\BO(n^2\max\{nm,\kappa\})$.
\end{proof}

\end{document}